\documentclass[letterpaper,11pt]{article}
\usepackage{verbatim}
\usepackage{amssymb, amsmath, appendix, verbatim, cancel, graphicx, wrapfig}
\usepackage{amsthm}
\usepackage{mathtools}
\usepackage{authblk}
\usepackage{textcomp, complexity}
\usepackage{fullpage}
\input{Qcircuit}

\newclass{\IQP}{IQP}
\newclass{\stoqQMA}{stoqQMA}
\newclass{\sampBPP}{sampBPP}
\newcommand{\PostIQPH}{\ComplexityFont{PostIQP}(H)}
\newcommand{\sampIQPH}{\ComplexityFont{samp\text{-}IQP}(H)}
\newclass{\PostBPP}{PostBPP}

\DeclareMathOperator{\diag}{diag}
\DeclareMathOperator{\Lie}{Lie}

\DeclareMathOperator{\Tr}{Tr}
\DeclareMathOperator{\spn}{span}

\DeclarePairedDelimiter{\set}{\lbrace}{\rbrace}
\DeclarePairedDelimiter{\abs}{\lvert}{\rvert}

\usepackage{color}

\begin{document}

\title{Complexity classification of two-qubit commuting hamiltonians}
\author[1]{Adam Bouland\thanks{email: adam@csail.mit.edu}}
\author[2]{Laura Man\v{c}inska\thanks{email: laura@locc.la}}
\author[1]{Xue Zhang\thanks{email: lzh@mit.edu}}
\affil[1]{Massachusetts Institute of Technology, Cambridge, MA USA}
\affil[2]{Centre for Quantum Technologies, National University of Singapore, Singapore}

\date{}

  \maketitle

\begin{abstract}
We classify two-qubit commuting Hamiltonians in terms of their computational complexity. 
Suppose one has a two-qubit commuting Hamiltonian $H$ which one can apply to any pair of qubits, starting in a computational basis state. 
We prove a dichotomy theorem: either this model is efficiently classically simulable or it allows one to sample from probability distributions which cannot be sampled from classically unless the polynomial hierarchy collapses.  
Furthermore, the only simulable Hamiltonians are those which fail to generate entanglement.
This shows that \emph{generic} two-qubit commuting Hamiltonians can be used to perform computational tasks which are intractable for classical computers under plausible assumptions. 
Our proof makes use of new postselection gadgets and Lie theory. 
\end{abstract}

\newtheorem{lemma}{Lemma}[section] 
\newtheorem{theorem}{Theorem}[section] 
\newtheorem{corollary}{Corollary}[section] 
\newtheorem{proposition}{Proposition}[section] 
\newtheorem{clm}{Claim}[section]
\newtheorem{definition}{Definition}[section]

\section{Introduction}

Quantum computers hold the promise of performing computational tasks which cannot be simulated efficiently using classical computers.
A hallmark example of this is Shor's quantum factoring algorithm \cite{Shor} for which no classical analog is known. 
However, proving that quantum computers hold an advantage over classical ones when it comes to factoring or any other decision problem would show that $\P\neq\PSPACE$, which is well beyond our current reach. 
Therefore, we aim to establish quantum advantage under widely accepted complexity assumptions like $\P\neq\NP$, non-collapse of the polynomial hierarchy $\PH$, and others.
In this submission we show that generic two-qubit commuting Hamiltonians can be used to perform computational tasks which are intractable for classical computers unless $\PH$ collapses. 
Since commuting gate sets allow for easier fault-tolerant implementation~\cite{experiment}, our results offer the possibility to experimentally perform classically intractable computations even before achieving universal quantum computation.

\subsection{Problem statement and results}

The evolution of a quantum system is determined by its Hamiltonian, which corresponds to a Hermitian matrix $H$. 
If we apply a Hamiltonian for time $t$, then this applies the unitary gate $e^{iHt}$ to the system.
The Hamiltonian of a system is governed by its underlying physics, so oftentimes in quantum computing experiments (e.g. in superconducting qubits) it is easy to apply certain Hamiltonians but not others.
From this perspective it is natural to study the computational power of a fixed Hamiltonian $H$  that can be applied to different ordered subsets of qubits for arbitrarily chosen amounts of time.  
Here we consider the model where
we have a fixed two-qubit\footnote{One-qubit Hamiltonians cannot create entanglement, so are efficiently classically simulable in this model.} Hamiltonian $H$ which we can apply to any ordered pair of qubits, where we initialize our system in a computational basis state and perform a computational basis measurement at the end. 
Now it is natural to ask: What is the computational power of this model for a fixed $H$? 
It is known that almost any choice of $H$ in this model yields universal quantum computation \cite{Deutsch,Weaver,2qubithamiltonians,Bauer}, but the classification of such universal Hamiltonians remains an open problem.
Curiously, there exist subsets of Hamiltonians that do not seem to offer the full power of $\BQP$ but nevertheless are hard to simulate classically under plausible complexity assumptions~\cite{IQP, IQPPH, shepherdIQPthesis}.

We focus on a particular family of Hamiltonians $H$ which, even though incapable of universal quantum computation, can perform computations that are hard for classical computers and might offer easier experimental implementation.
Specifically, we study Hamiltonians $H$ that can only give rise to mutually commuting gates:

\begin{definition}
We say that a two-qubit Hamiltonian $H$ is \emph{commuting} if $[H\otimes I, I \otimes H]=0$ and $[H\otimes I, I\otimes (THT)]=0$ and $[H, THT]=0$, where $T$ is the two-qubit swap gate. In other words, $H$ commutes with itself when applied to any pair of qubits.
\end{definition}

We are interested in classifying which commuting two-qubit Hamiltonians $H$ allow us to perform computational tasks that are hard for classical computers. In particular, we want to understand when $H$ gives rise to probability distributions which are hard to simulate classically:
\begin{definition}
We say that a family of probability distributions $\{\mathcal{D}_x\}_{x\in\{0,1\}^*}$ are \emph{hard to sample from classically} if there exists a constant $c>1$ such that no $\BPP$ machine $M$ can satisfy
\[
	\frac{1}{c}\Pr[\text{$M(x)$ outputs $y$}] \leq \mathcal{D}_x(y) \leq c\Pr[\text{$M(x)$ outputs $y$}]
\]
for all $y$ in the sample space of $\mathcal{D}_x$. 
\end{definition}

Clearly, if a commuting $H$ is not capable of creating entanglement from any computational basis state then the system will remain in a product state, so this model will be efficiently classically simulable. 
Surprisingly, we show that in all the remaining cases $H$ can perform sampling tasks which cannot be simulated classically unless $\PH$ collapses.

\begin{theorem}[Main Result]
If a commuting two-qubit Hamiltonian $H$ is capable of creating entanglement from a computational basis state, then it gives rise to probability distributions that are hard to sample from classically unless $\PH$ collapses.
\label{thm:Main}
\end{theorem}

Additionally, given such an $H$, our result provides an algorithm which describes the experimental setup required to sample from these hard distributions.

\paragraph{Experimental implications.}
Universal quantum computers have proved challenging to implement in practice as they require large overheads for fault-tolerance.
As a result, some skeptics have questioned if quantum devices will ever be able to demonstrate an advantage over classical computers \cite{Kalai,Levin}.

One response to this challenge is to study weaker models of quantum computation which are incapable of universal quantum computation, but still demonstrate an advantage over classical computation \cite{Knill1998,DQC2014,Jordan,bosonsampling,IQPPH}. 
Aliferis et al.\ \cite{experiment} have shown that commuting gate sets may be easier to implement fault-tolerantly with superconducting qubits than universal gate sets, and provided numerical evidence that they may admit lower fault-tolerance thresholds.
Therefore, commuting computations form a good candidate for providing the first ``proof of concept" demonstration of quantum supremacy over classical computation \cite{bosonsampling}. Our Theorem~\ref{thm:Main} says that almost any commuting Hamiltonian could be used for this demonstration, and additionally provides the experimentalist with a straightforward criterion to determine whether a commuting Hamiltonian can be used to sample from hard distributions.

\subsection{Proof ideas}
\vspace{-0.5em}
Our proof proceeds in several steps. 
First, we use that fact that any commuting two-qubit Hamiltonian $H$ is locally diagonalizable:
\begin{lemma}[\cite{cubitt} Lemma 33]
For any commuting two-qubit Hamiltonian there exists a one-qubit unitary $U$ and a diagonal matrix $D$ such that $H=(U\otimes U) D (U^\dagger \otimes U^\dagger$).
\end{lemma} 
The proof of this follows from expanding $H$ in the Pauli basis, and  deducing relationships between the Pauli coefficients.

Next, we use postselection gadgets to construct a family of one-qubit operations $L(t):\mathbb{C}^2\rightarrow\mathbb{C}^2$ for $t\in\mathbb{R}$ that that can be applied to the input state using postselection. 
We then show that these gadgets are universal on a qubit whenever $H$ generates entanglement, so long as $H$ is not some exceptional subcase.
The exceptional subcase is $H=X(\theta)\otimes X(\theta)$ where $X(\theta)=\left(\begin{smallmatrix}0&e^{i\theta/2} \\ e^{-i\theta/2} & 0 \end{smallmatrix}\right)$. 

\begin{lemma} 
If $H$ is capable of creating entanglement from a computational basis state and $H$ is not $X(\theta)\otimes X(\theta)$ for some $\theta$, then it is possible to construct any one-qubit gate by taking products of the $L(t)$ gadgets.
\end{lemma}
The main difficulty in proving this fact is that the maps $L(t)$ are in general  \emph{non-unitary}. Furthermore since they are generated with postselection, it is unclear how to invert them, so \emph{a priori} they might not even form a group. Fortunately, we find new (and somewhat complicated) postselection gadgets to construct the $L^{-1}$ operations, thus allowing us to apply group-theoretic and  Lie-theoretic techniques to address this problem.

The rest of the proof follows from standard techniques in complexity. 
Since one-qubit gates plus any entangling Hamiltonian form a universal gate set \cite{anyhamiltonianuniversal,anygateuniversal}, our model can perform universal quantum computation under postselection. 
\begin{lemma}
If $H$ is capable of creating entanglement from a computational basis state and $H$ is not $X(\theta)\otimes X(\theta)$ for some $\theta$, then postselected circuits involving $H$ are universal for $\BQP$.
\end{lemma}
The proof of this statement uses a non-unitary version of the Solovay-Kitaev theorem proven by Aharonov et al. \cite{aharonov} to show our choice of  gate set is irrelevant. 

Next, a result of Aaronson \cite{postbqp} tells us that postselecting our circuits further enables us to solve $\PP$-hard problems.
It then follows by the complexity arguments put forth by Bremner, Jozsa, and Shepherd \cite{IQPPH} and Aaronson \cite{bosonsampling} that a randomized classical algorithm cannot sample from the probability distributions produced by our circuits unless the polynomial hierarchy collapses.

This completes the classification for all cases except the case $H=X(\theta)\otimes X(\theta)$. Hardness of sampling from these Hamiltonians was previously shown by Fefferman, Foss-Feig, and Gorshkov \cite{feffermanpersonal} using a construction which embeds permanents directly in the output distributions of such Hamiltonians. Hardness then follows from the arguments of Aaronson and Arkhipov \cite{bosonsampling}. We provide a summary of their hardness result for completeness.

\subsection{Relation to prior work}

Our work is inspired by Bremner, Jozsa, and Shepherd \cite{IQP, IQPPH, shepherdIQPthesis}, who showed that  certain computations with commuting gates are hard to simulate classically unless the polynomial hierarchy collapses. 
In particular, they show hardness of simulating the gate set comprised of $HZH$, $(H \otimes H)(\text{controlled-}Z)(H \otimes H)$, and $HPH$, where $P$ is the $\pi/8$-phase gate.
Similarly, Shepherd \cite{shepherdmatroids,shepherdIQPthesis} considers the power of applying quantum Hamiltonians which are diagonal in the $X$ basis, where the Hamiltonians can be applied only for discrete amounts of time $\theta$; he describes the values of $\theta$ for which the resulting circuits are efficiently classically simulable or hard to weakly simulate (that is, to sample from the output probability distribution with a classical computer).
Our work differs from these in several ways. First, We consider Hamiltonians rather than gates, and show hardness of \emph{generic} or \emph{average-case} commuting Hamiltonians, rather than showing hardness for worst-case commuting operations. 
Furthermore, we fully classify the computational complexity of all commuting Hamiltonians, and prove a dichotomy between hardness and classical simulability.

The hardness results we obtain in this paper (as well as those in \cite{IQPPH,shepherdmatroids,shepherdIQPthesis}) are based on the difficulty of sampling the output probability distribution on all $n$ output qubits.  
A number of other works have considered the power of computations with commuting Hamiltonians, where one only considers the output distribution on a small number of output qubits.
For example, Bremner, Jozsa and Shepherd \cite{IQPPH} showed that computing the  marginal probability distributions on $O(\log(n))$ qubits of their model is in $\P$.
Ni and Van den Nest \cite{NiNest} showed that this holds for arbitrary 2-local commuting Hamiltonians, but also showed there exist 3-local commuting Hamiltonians for which this task is hard. 
Hence the problem of strongly simulating the output distributions  (that is, being able to compute the probability of any event) of arbitrary $k$-local Hamiltonians is hard for $k\geq 3$. 
Along a similar line of thought, Takahashi et al.\ \cite{Takahashi} showed that there is a system of 5-local commuting Hamiltonians for which weakly simulating the output on $O(\log(n))$ bits is hard.

Additionally, a number of other authors have also considered ``weak" models of quantum computation which can sample from difficult probability distributions. Some examples include the one clean qubit model \cite{Knill1998,DQC2014}, the boson sampling model \cite{bosonsampling}, the quantum fourier sampling model \cite{fefferman}, and temporally unstructured quantum computing \cite{IQPPH}. 
Like many of these models (e.g. \cite{DQC2014,IQPPH}), we prove it is difficult for a classical computer to sample from the distribution output by the quantum device with multiplicative error  on every output probability. 
For some of these models, the authors prove stronger hardness results for sampling the output distribution with additive error (as measured in trace distance) \cite{bosonsampling, IQPadditive, fefferman}, but at the cost of making additional complexity-theoretic assumptions which are not as widely accepted. 
In comparison with boson sampling, one clean qubit sampling, and quantum fourier sampling, our model has the advantage of possibly having lower fault-tolerance thresholds for implementation \cite{experiment}.

Finally, other works have addressed the classification of universal two-qubit gates and Hamiltonians. 
Childs, Leung, Man\v{c}inska, and  Ozols \cite{2qubithamiltonians} classified the set of two-qubit Hamiltonians which give rise to $SU(4)$ when acting on two qubits, and are hence universal. 
Lloyd \cite{Lloyd2qubit} and others \cite{Deutsch,Weaver,2qubithamiltonians,Bauer} have shown that a Haar-random two-qubit gate is universal with probability 1.
Our work differs from these in that our Hamiltonians only become universal under postselection. Additionally, Cubitt and Montanaro \cite{cubitt} previously classified the complexity of two-qubit Hamiltonians in the Local Hamiltonian Problem setting. Specifically, given a two qubit Hamiltonian $H$, they classify the computational complexity of determining the ground state energy of Hamiltonians of the form $\sum_{ij} c_{ij}H_{ij}$ for real coefficients $c_{ij}$. This is incomparable with our classification, since we are studying the power of the Hamiltonian dynamics (in which the system is not in the ground state), rather than the complexity of their ground states.

\section{Preliminaries and statement of Main Theorem}

A two-qubit Hamiltonian $H$ is a $4\times 4$ Hermitian matrix. 
Let $T$ denote the SWAP gate. 
Given $H$, we assume that one can apply either $H$ or $THT$ to any pair of qubits. 
In other words, we can apply the Hamiltonian oriented from qubit $i$ to qubit $j$, or from qubit $j$ to qubit $i$. 
We will use $H_{ij}$ to denote the Hamiltonian applied from qubit $i$ to qubit $j$. 
Additionally, we will assume we can apply $-H$ as well, i.e., we can perform the inverse Hamiltonian.\footnote{If we had only assumed access to $H$ and positive time evolution, we could always approximate the action of $-H$; this follows from compactness of the unitary group and was shown e.g.\ in Appendix A of \cite{2qubithamiltonians}. However, here we are assuming we have exact access to $-H$; this will be useful when arguing about post-selected versions of these circuits.}

Suppose we are given some input string $x\in\{0,1\}^n$, and we want to define a distribution on $n'=\poly(n)$ bits which we can efficiently sample from using $H$. 
Suppose we initialize a system of $n'$ qubits in a computational basis state $\ket{y}$ for $y\in\{0,1\}^{n'}$, apply each Hamiltonian $H_{ij}$ for time $t_{ij}\in\mathbb{R}$, and then measure all the qubits in the computational basis.
(Here the times $t_{ij}$ and the string $y$ may depend on $x$.) 
This will induce some probability distribution $\mathcal{D}_x$ over bit strings of length $n'$ on the output bits.
Intuitively, these are the sorts of distributions one can efficiently sample from using $H$, using circuits which start and end in the computational basis.

However, this definition does not quite suffice to capture a realistic model of computation, because we have not specified how the initial state $y$ and the times $t_{ij}$ are chosen. 
To fix this, we will require that one could use a classical computer to efficiently calculate the experimental setup for each $n$. In other words, we will require that there exists a polynomial-time algorithm which, given $x\in \{0,1\}^*$, computes the values of $y$ and $t_{ij}$ used in the experiment. 
Furthermore, we will require that the times $t_{ij}$ can be represented with polynomially many bits, and that they are all bounded in magnitude by a polynomial in $n$.
This ensures that as the size of the system grows, the amount of time one needs to run the Hamiltonian does not grow too quickly.
In complexity theory this is called a \emph{uniformity} condition. This requirement ensures that any advantage over classical computation arising from this model comes from the power of the quantum computation performed, not the computation of the experimental setup.

This is stated more formally as follows:
\begin{definition} Let $\sampIQPH$ denote those families of probability distributions $\{\mathcal{D}_x\}$ for which there exists a classical poly-time algorithm $\mathcal{A}$ which, given an input $x\in \{0,1\}^n$, outputs the specifications for a quantum circuit using $H$ whose output distribution family is $\{\mathcal{D}_x\}$. In particular, $\mathcal{A}$ specifies a number of qubits $n'=\poly(n)$, a string $y\in\{0,1\}^{n'}$ and and a series of times $t_{ij}\in\mathbb{R}$, such that running a quantum circuit starting in the state $\ket{y}$, applying the operator $e^{it_{ij}H_{ij}}$ for each $(i,j)$, and then measuring in the computational basis will yield a sample from $\mathcal{D}_x$. Each $t_{ij}$ must be specifiable with $\poly(n)\) bits and be bounded in magnitude by a polynomial in \(n\).
\end{definition}

In short, the class $\sampIQPH$ captures the set of probability distributions one can efficiently sample from using $H$. In our work, we will show that a classical randomized algorithm cannot sample from this same set of distributions. More precisely, we say that a classical randomized algorithm ``weakly simulates" a quantum circuit if its output distribution is close to the output distribution of the quantum circuit. To derive our hardness result, we will consider classical circuits which produce every output with approximately the correct probability, up to multiplicative error: 
\begin{definition} A $\BPP$ (bounded-error probabilistic polynomial time) machine $M$ \emph{weakly simulates} a family of probability distributions $\{P_x: x\in\set{0,1}^*\}$, where $P_x$ is a distribution over $\{0,1\}^{\abs{x}}$, with multiplicative error $c\geq 1$ if, for all $y\in \{0,1\}^n$,
\[
	\frac{1}{c}\Pr[\text{$M(x)$ outputs $y$}] \leq P(x) \leq c\Pr[\text{$M$ outputs $y$}].
\]
\end{definition}

We can now more precisely state our Main theorem: that our commuting circuits cannot be weakly simulated unless the polynomial hierarchy $\PH$ collapses:

\begin{theorem}[Main Theorem] If $H$ is capable of generating entanglement from the computational basis, then $\BPP$ machines cannot weakly simulate $\sampIQPH$ with multiplicative error $c<\sqrt{2}$ unless $\PH$ collapses to the third level. \label{maintheorem}
\end{theorem}

In other words, there is a dichotomy: either computations which $H$ are efficiently classically simulable, or else they cannot be efficiently simulated unless the polynomial hierarchy collapses. As the non-collapse of the polynomial hierarchy is a widely accepted conjecture in computational complexity, this is strong evidence that $\sampIQPH$ circuits are not efficiently classically simulable.

\subsection{Complexity Preliminaries}

Before proceeding to a proof of the Main Theorem, we will introduce some of the complexity-theoretic preliminaries necessary to understand our proof. We assume the reader is familiar with the standard complexity classes such as $\P$, $\BPP$, and $\NP$, as well as oracle notation; for background we refer the reader to Arora and Barak \cite{complexity} for details. Those readers already familiar with the complexity theoretic techniques of Bremner, Jozsa, and Shepherd \cite{IQPPH} and Aaronson and Arkhipov \cite{bosonsampling} may wish to skip to the proof of the Main Theorem.

In order to reason about the computational complexity of $\sampIQPH$ distributions, we will need to introduce the idea of postselected circuits. 
A postselected quantum circuit is a circuit where one specifies the value of some measurement results ahead of time, and discards all runs of the experiment which do not obtain those measurement outcomes. 
This is not something one can realistically do in a laboratory setting, because the measurement outcomes you specify may occur extremely infrequently---in fact, they may be exponentially unlikely. 
However, postselection can help you examine the \emph{conditional probabilities} found in the output distribution of your circuit.
In particular, if you can show that those conditional probabilities can encode the answers to very difficult computational problems, then this can provide evidence against the ability to classically simulate such circuits. 
Therefore, we will now define what it means for a set of probability distributions to decide a problem under postselection. 
The basic idea is that if some of the conditional probabilities of the system encode the answer to a problem, then we say that problem can be decided by postselected versions of these probability distributions. 
We define this more formally below:

\begin{definition} 
Let $\PostIQPH$ be the set of languages $L\subseteq \{0,1\}^*$ for which there exists a family of $\sampIQPH$ circuits $\{\mathcal{D}_x\}$ and a classical poly-time algorithm which, given an input length $n$, outputs a subset $B$ of qubits and a string $z\in\{0,1\}^{|B|}$ such that 
\begin{itemize}
	\item If $x\in L$, then $\Pr[\text{ $\mathcal{D}_x$ outputs 1 on its first bit } | \text{ bits $B$ take value $z$ }] \geq 2/3$.
	\item If $x\notin L$, then $\Pr[\text{ $\mathcal{D}_x$ outputs 1 on its first bit } | \text{ bits $B$ take value $z$ }] \leq 1/3$.
\end{itemize}
\end{definition}
In other words, there exists a poly-time algorithm which outputs an experimental setup and a postselection scheme such that the conditional probabilities of $\mathcal{D}_x$ encode the answer to the problem.
In general, the choice of  constants $1/3$ and $2/3$ in the above definition might matter. 
For instance, when $\PostIQPH$ is not capable of universal classical computation, it is unclear how to take the majority vote of many repetitions to amplify the success probability. 
However, we only consider the class $\PostIQPH$ in cases where $\PostIQPH$ can perform universal classical computation, and thus the choice of constants $1/3$ and $2/3$ is arbitrary.

One can likewise define the classes $\PostBQP$ and $\PostBPP$\footnote{ $\PostBPP$ is more commonly known as $\BPPpath$.} which capture the power of postselected quantum computation and postselected randomized computation, respectively.

Finally, we introduce the polynomial hierarchy. The $i$th level of the polynomial hierarchy, denoted $\Delta_i$, is defined as follows: let $\Delta_1 = \P$, let $\Delta_2=\P^\NP$, let $\Delta_3 = \P^{\NP^\NP}$, let $\Delta_4 = \P^{\NP^{\NP^{\NP}}}$, and so on. 
Here, we write $A^B$ to refer to computations that can be performed with an $A$ machine which has been augmented with the ability to solve problems in $B$ in a single timestep. 
The polynomial hierarchy, denoted $\PH$, is defined as $\PH = \bigcup_{i\in\mathbb{N}} \Delta_i$. 
It is widely conjectured that each level of the polynomial hierarchy is distinct; in other words, $\Delta_i \subsetneq \Delta_{i+1}$ for all $i\in\mathbb{N}$. This can be seen as a generalization of the conjecture that $\P\neq\NP$.

One of the main technical tools we will use in our proof is the following lemma, which was first shown by Bremner, Jozsa, and Shepherd \cite{IQPPH}, but which we will make extensive use of in our paper:
\begin{lemma} \label{hardnesslemma}
Suppose that $\PostBQP \subseteq \PostIQPH$ for some $H$. 
Then $\BPP$ machines cannot weakly simulate $\sampIQPH$ with multiplicative error $c<\sqrt{2}$ unless $\PH$ collapses to the third level.
\end{lemma}

In other words, if postselected commuting Hamiltonian circuits are capable of performing (postselected) universal quantum computation, then they cannot be weakly simulated by a classical computer under plausible complexity assumptions. 
The fundamental reason this is true is that the class $\PostBQP$ is substantially more powerful than the class $\PostBPP$. 
In particular, Aaronson \cite{postbqp} showed that $\PostBQP=\PP$. Here $\PP$ (which stands for Probabilistic Polynomial-time) is the set of languages $L$ decidable by a poly-time randomized algorithm $M$, such that 
\begin{itemize}
	\item If $x\in L$, then $\Pr[M(x) \text{ accepts}] > 1/2$;
	\item otherwise, $ \Pr[M(x) \text{ accepts}] \leq 1/2$.
\end{itemize}
In other words, the class $\PP$ represents the class of problems solvable by randomized algorithms, where the probability of acceptance of ``yes" and ``no" instances is different, but may only differ by an exponentially small amount\footnote{Note the difference is probabilities is always at least $2^{-\text{poly}(n)}$, because a $\PP$ algorithm can only make polynomially many coin flips.}. 
A famous result in complexity, known as Toda's Theorem \cite{Toda}, states that $\PH \subseteq \P^\PP$. 
In other words, the class $\PP$ is nearly as powerful as the entire polynomial hierarchy.

On the other hand, the class $\PostBPP$ is far weaker; it lies in the third level of the polynomial hierarchy.
So if one assumes that $\PH$ does not collapse to the third level, then $\PostBPP \neq \PostBQP$; i.e. $\PostBQP$ is a stronger complexity class than $\PostBPP$.

From this, we can now state why the inclusion $\PostBQP \subseteq \PostIQPH$ implies there cannot exist an algorithm to simulate $\PostIQPH$ circuits. Suppose there were a $\BPP$ algorithm to weakly simulate such circuits. Then, by postselecting this $\BPP$ algorithm, we could solve a $\PostBQP$-hard problem in $\PostBPP$, which would imply the collapse of the polynomial hierarchy. A more formal statement of this proof is given below:

\begin{proof}[Proof of Lemma \ref{hardnesslemma}] The proof of this corollary is given in \cite{IQPPH} Theorem 2 and Corollary 1, but we provide a summary for completeness. 
Suppose that a $\BPP$ machine $M$ can weakly simulate $\sampIQPH$ circuits to multiplicative error $c<\sqrt{2}$. 
Then for any individual output string $x$, we have  $\frac{1}{c}\Pr[M\text{ outputs }x] \leq P(x) \leq c\Pr[M\text{ outputs } x]$. Since $\PostBQP \subseteq \PostIQPH$, and $\PostBQP=\PP$ \cite{postbqp}, this can be shown to imply $\PP \subseteq \PostBPP$.
But $\PostBPP\subseteq \PostBQP = \PP$, so this implies $\PostBPP=\PP$.
Hence by Toda's theorem \cite{Toda}, we have $\PH \subseteq \P^\PP= \P^\PostBPP \subseteq \Delta_3$, where $\Delta_3$ is the third level of the polynomial hierarchy. Hence $\PH=\Delta_3$ as claimed.
\end{proof}

Note that in certain cases, Fujii et al. \cite{Fujii2014} showed that this hardness of simulation result could be improved to imply the collapse of $\PH$ to the second level rather than the third, using a different complexity-theoretic argument involving the class $\NQP$. 
However, their argument is gate-set dependent, so does not apply to our model for arbitrary commuting Hamiltonians.

We now proceed to a proof of the Main Theorem.

\section{Proof of Main Theorem}

The basic idea is to use postselection gadgets to show that postselected $\sampIQPH$ circuits are capable of performing universal quantum computation. 
Hence, adding further postselections allows one to decide any language in $\PostBQP$. 
By Lemma \ref{hardnesslemma}, this proves hardness of weakly simulating such circuits unless $\PH$ collapses.

\begin{proof}[Proof of Theorem \ref{maintheorem}]

Suppose we have a commuting two-qubit Hamiltonian $H$. The first step in our proof is to characterize the structure of such $H$.
It is clear that if $H$ is diagonal under a local change of basis, i.e. $H=(U\otimes U) D (U^\dagger  \otimes U^{\dagger})$ for some one-qubit $U\in SU(2)$ and diagonal matrix $D$, then $H$ is commuting. 
However, it is possible \emph{a priori} that there exist commuting Hamiltonians which are not of this form. 
If $T$ is the gate that swaps two qubits, then the fact that $H$ is commuting implies that $H\otimes I$, $(THT)\otimes I$, $I\otimes H$, and $I\otimes (THT)$ are all simultaneously diagonalizable. However, it might be that this simultaneous diagonalization can only happen under a non-local change of basis.
Fortunately, it turns out this is not possible - any commuting Hamiltonian must be locally diagonalizable.
This was first shown by Cubitt and Montanaro \cite{cubitt}. 
\begin{clm}[\cite{cubitt} Appendix B Lemma 33]\label{clm:localdiag} If $H$ is a 2-local commuting Hamiltonian, then $H=(U\otimes U) D (U^\dagger \otimes U^\dagger)$ for some one-qubit $U\in SU(2)$ and diagonal matrix $D$.
\end{clm}
We provide a proof of Claim \ref{clm:localdiag} in Appendix \ref{app:localdiag}, which uses expansion in the Pauli basis. 
One can also prove this fact using linear algebra, but the proof becomes  complicated in the case of degenerate eigenvalues. 
We thank Jacob Taylor for pointing us to this simplified proof, and Ashley Montanaro for pointing us to the proof in reference \cite{cubitt}.

By Claim \ref{clm:localdiag}, we know that $H = (U\otimes U) \diag (a,b,c,d) (U^\dagger \otimes U^\dagger)$ for some one-qubit unitary $U=\left(\begin{smallmatrix} \alpha & -\beta^* \\ \beta & \alpha^*\end{smallmatrix}\right)$ and some real parameters $a,b,c,d$.
The trace of $H$ contributes an irrelevant global phase to the unitary operator it generates, so without loss of generality we can assume $H$ is traceless, i.e., $a+b+c+d=0$.

Note that if $a=d=-1$, $b=c=1$, and $|\alpha|=|\beta|$, then we have that $H=X(\theta)\otimes X(\theta)$, where $e^{i\theta}=\alpha/\beta$. As mentioned previously, these Hamiltonians are hard to simulate by an independent hardness result of Fefferman \emph{et al.} \cite{feffermanpersonal}, so in the rest of our proof, we will assume we are not in the case $a=d=-1$, $b=c=1$ and $|\alpha|=|\beta|$.  For completeness we will provide a summary of their work at the end of this proof.

We now consider the conditions under which computations with $H$ are efficiently classically simulable. First, if $H$ is diagonal in the computational basis, then it is obviously classically simulable, because it cannot generate entanglement from the computational basis. This corresponds to the case that $\alpha=0$ or $\beta=0$. So we can assume for the result of the proof that $\alpha\neq 0$ and $\beta \neq 0$.

Another way that $H$ can fail to generate entanglement from the computational basis is if $b+c=a+d$. Since we are assuming the Hamiltonian is traceless this is equivalent to the condition $b+c=0$. 
Indeed if $H$ satisfies $b+c=0$, and $H$ is traceless so $a+d=0$, then it is easy to check that $e^{iHt}$ is nonentangling for all $t\in\mathbb{R}$. So we can assume in the rest of the proof that $b+c\neq 0$. 

We now show that for all remaining $H$, we have $\PostBQP\subseteq \PostIQPH$. 
To do so, we break into two cases. 
Either $b=c$, so $H=THT$ and the Hamiltonian is identical when applied from qubit 1 to 2 vs. from 2 to 1, or $b\neq c$ so $H\neq THT$. 
For clarity of presentation, we will prove our main theorem in the case $b=c$, as this proof uses simpler notation.
An analogous proof holds for the case $b\neq c$, which we provide in Appendix \ref{app:edgecase}.

Now in the case $b=c$, consider the rescaled Hamiltonian $H'=H/b$. Since $b+c\neq0$ and $b=c$ this Hamiltonian is well-defined, and we have $H'=(U\otimes U) \diag (a',1,1,d') (U^\dagger \otimes U^\dagger)$
for some real parameters $a'$ and $d'$ which obey $a'+d' =-2$.
Now consider the two-qubit unitary $V(t)$ we obtain from running $H'$ for time $t\in \mathbb{R}$ 
\[
V(t) = e^{itH'} = (U^{\otimes 2}) D(t) ({U^\dagger}^{\otimes 2}),
\]
where $D(t)\triangleq \diag(e^{ia't},e^{it},e^{it},e^{id't})$. Here we have used the fact that if $U$ is an arbitrary unitary, then $e^{UH'U^\dagger} = U e^{H'} U^\dagger$.

A $\sampIQPH$ circuit is specified by times $t_{ij}$ for all unordered\footnote{This is because we are considering the case $b=c$ i.e. $H=THT$.} pairs of qubits $(i,j)$, as well as an initial basis state $\ket{y}$ for $y\in\{0,1\}^{\poly(n)}$. The circuit consists of applying $ V(t_{ij})$ to each pair of qubits $(i,j)$ to $\ket{y}$, and then measuring in the computational basis. 
This can be easily seen to be equivalent to the following circuit: Start in the state $\ket{y}$, apply $U$ to every qubit, then apply $D(t_{ij})$ to each pair of qubits; finally, apply $U^\dagger$ to every qubit and measure in the computational basis. 
(This is true because all factors of $U$ and $U^\dagger$ in the circuit cancel except those at the beginning and end). 

We will now show how to make post-selected gates of this form perform universal quantum computing. 
The basic idea is that we already have a two-qubit entangling Hamiltonian at our disposal. Therefore, if we could show how to perform arbitrary one-qubit gates using post-selection, this would form a universal gate set for quantum computing by the result of Dodd et al.\ \cite{anyhamiltonianuniversal} or Bremner et al.\ \cite{anygateuniversal}. Following the method of Bremner, Jozsa, and Shepherd \cite{IQPPH}, we consider the following post-selection gadget, denoted $L(t)$, which performs an operation on a single qubit state $\ket{\psi}$:
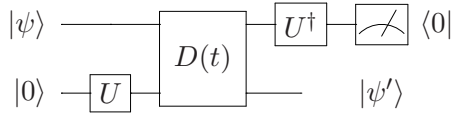
\[\Qcircuit @C=1em @R=1em {
\lstick{\ket{\psi}} & \qw & \multigate{1}{D(t)} & \gate{U^\dagger} & \meter & \bra{0}   \\
\lstick{\ket{0}} & \gate{U} & \ghost{D(t)} & \qw & \ket{\psi'} 
}
\]
The postselection is denoted in the circuit by $\bra{0}$. 
Note that this gadget preserves the property that every line begins with $U\ket{0}$, and ends with $U^\dagger$ and a measurement. 
Hence, if we could use these postselection gadgets to perform arbitrary single-qubit gates, then we could perform universal quantum computing under postselection as follows: Given a target quantum circuit to simulate, compile the circuit out of gates of the form $D(t)$ and single-qubit gates.
Additionally, add a $UU^\dagger$ (which is the identity) at the beginning and end of every line, so that each line starts with a $U$ and ends with a $U^\dagger$. 
Now this circuit consists of applying a column of $U$'s, then a series of diagonal gates $D(t)$ and one-qubit gates, followed by a column of $U^\dagger$s. 
This almost has the form of a $\sampIQPH$ circuit, with the exception of the one-qubit gates (note that these include both the gates $U^\dagger$ in the second column and the gates $U$ in the second to last column). 
Now for each one-qubit gate $g$, replace it with its implementation using postselection gadgets $L(t)$.
After this transformation, each line begins with a $U$, ends with a $U^\dagger$, and contains only diagonal gates $D(t)$ in the interior of the circuit.
However, now we've additionally specified some postselection bits, so we have created a $\PostIQPH$ circuit which simulates universal quantum computing.

Let us examine what transformation $L(t)$ actually performs on the qubits involved.  The gadget performs some linear transformation on the input state $\ket{\psi}$. In particular, it acts on $\ket{\psi}$ by
\[L(t) = \frac{1}{|\alpha||\beta|\sqrt{-2i\sin(2t)}}  \left( \begin{matrix} |\alpha|^2 e^{ia't}& \alpha \beta^* e^{it} \\ \alpha^* \beta e^{it} & |\beta|^2 e^{id't)} \end{matrix} \right).\]
This is a non-unitary transformation, so it does not preserve the norms of vectors.
Since we only care about how \(L(t)\) behaves on the projective Hilbert space of quantum states, we can choose the overall normalization so that $L(t) \in SL(2, \mathbb{C})$. 
Note that this operator is well-defined only if the denominator above is non-zero, so we will require that $t \in (0,\pi)\cup(\pi,2\pi)$.

In addition to being able to perform the transformation $L(t)$ as 
$t$ ranges over  $t \in (0,\pi)\cup(\pi,2\pi)$, we can also perform products of such transformations.
In fact, we can perform any operation in the set
\[
S \triangleq \overline{ \left\langle \left\{ L(t) : t \in  (0,\pi)\cup(\pi,2\pi)   \right\} \right\rangle }.
\] 
Here the angled brackets $\left\langle A \right\rangle$ denote the set of all matrices obtained by finite products of elements of $A$. 
The bar above $\overline{\langle A \rangle}$ means that we take the closure of this set in $SL(2,\mathbb{C})$; in other words, we include all matrices that one can obtain by taking limits of sequences of finite products of $A$, so long as the limit point belongs to $SL(2,\mathbb{C})$.

If the matrices $L(t)$ were in a compact space such as $SU(2)$, then it would immediately follow that $S$ contains inverses of all its elements.\footnote{To see this, take an element $s\in S$. If $s$ has finite order, than its inverse is clearly in $S$. If $s$ has infinite order, consider the sequence $1, s, s^2, \ldots$. Since the matrices are in a compact space $T$, the sequence of powers must have a convergent subsequence, i.e.\ there must be positive $n_1,n_2,n_3 \ldots$ such that $n_1 < n_2 < \ldots$ and $s^{n_1},s^{n_2}, \ldots$ approach some element $t \in T$. Therefore the sequence $s^{n_2-n_1}, s^{n_3-n_2}, \ldots$ must approach the identity, and the sequence $s^{n_2-n_1-1}, s^{n_3-n_2-1}, \ldots$ must approach $s^{-1}$.}
Therefore we would know that $S$ is a group, and we could apply tools from group theory to categorize $S$. 
However, our matrices are in the non-compact space $SL(2,\mathbb{C})$. 
Therefore it is not clear whether $S$ is closed under taking inverses, so $S$ \emph{might not be a group}! 
Furthermore, since $L$ is obtained under post-selection, the assumption that we can perform the inverse of $H$ does not imply we can perform $L^{-1}$.

To fix this problem, we find additional gadgets which allow us to construct $L^{-1}$ by adding additional postselections to our circuit. In particular, we will show that for each $L(t)$, there exists a postselection gadget of finite size which performs $L(t)^{-1}$ exactly. An important restriction on this construction is that this inverse must be efficiently computable. 
Specifically, for each $L(t)$ the size of the postselection gadget required to invert $L(t)$ is of constant size.
Additionally, the construction of the postselection gadget will in general contain several time parameters which one needs to set in order to obtain $L(t)^{-1}$. We also require that we can set these times so that we obtain $L^{-1}$ to accuracy $\epsilon$ in poly$\log(k_L 1/\epsilon)$ time, where $k_L$ is a constant which depends on $L(t)$ only. 
Furthermore, the amount of time needed to run the Hamiltonians in the inverse gadget are bounded above by a polynomial.
For convenience we will refer to these properties as ``the construction is efficiently computable."

At first glance it might sound like this definition of ``efficiently computable" is too weak, because the inverses of arbitrary $L$ matrices might require large postselection gadgets. 
However, later in our construction we will use the fact that for any fixed Hamiltonian $H$, we will only need to invert a finite set of $L$ matrices.
Hence for fixed $H$, the size of the postselection gadgets which appear in our circuit will be upper bounded by a constant depending on $H$ only, but not on the size of the problem we are solving under postselection. Furthermore, for fixed $H$, we can compute the times in the inversion gadgets to invert the relevant $L$ matrices to exponential accuracy in polynomial time. This ability to invert the $L$ matrices to exponential precision will later be crucial for our hardness of sampling result.

Furthermore, note that in the case that $H\neq THT$, the construction of these gadgets can be made substantially simpler. 
In particular, the gadgets to construct $L^{-1}(t)$ are of size 4 for any $t$, and the times used in running the Hamiltonians are trivially efficiently computable to polynomial digits of accuracy. 
From a practical experimental perspective these circuits would be easier to construct, and since $H\neq THT$ is the generic case for commuting Hamiltonians, would be applicable for almost all commuting Hamiltonians. We include this construction in Appendix \ref{app:edgecase}.

\begin{clm} \label{clm:inverses} For any given $L(t)$, where $t\in (0,\pi)\cup(\pi,2\pi)$, it is possible to construct $L(t)^{-1}$ by introducing a constant number of postselections and a constant number of ancillas into the circuit. Furthermore, this construction is efficiently computable in the manner described above.
\end{clm}

The proof of Claim \ref{clm:inverses} can be found in Appendix \ref{app:inverses}, and is somewhat involved.

We now redefine $S$ so that its base set contains these inverses:
\[
S \triangleq  \overline{ \left\langle\left\{ L(t) : t\in (0,\pi)\cup(\pi,2\pi)  \right\} \cup \left\{ L(t)^{-1} :t \in  (0,\pi)\cup(\pi,2\pi) \right \}\right\rangle }.
\]
Using this definition, we can now show using standard techniques that $S$ is a Lie group---this is essentially a consequence of Cartan's closed subgroup theorem \cite{Cartan1952} and the fact that inversion is a continuous operation in the matrix entries on $SL(2,\mathbb{C})$. 
Once we know that $S$ has the structure of a Lie group, we can apply the theory of Lie algebras to identify what set of matrices are in $S$. In particular, we can show that $S$ generates all of $SL(2,\mathbb{C})$.

\begin{clm}\label{density}  $S=SL(2,\mathbb{C})$.
\end{clm}
The proof of this claim is a tedious but straightforward calculation using Lie algebras and properties of the exponential map on $SL(2,\mathbb{C})$. 
The proof uses the fact that we are not in one of the cases excluded by our theorem (i.e. $H$ does generate entanglement and is not $X(\theta)\otimes X(\theta)$ for some $\theta$) - in these cases one does \emph{not} find that $S=SL(2,\mathbb{C})$ as one would expect.
In certain special cases, the gadgets $L(t)$ alone do not generate $SL(2,\mathbb{C})$, specifically when $a'=\pm1$ or $a'=-3$. In these cases, we show that one can add additional postselection gadgets, which are closed under taking inverses, which boost the power of the $L(t)$ transformations to cover all of $SL(2,\mathbb{C})$. This simply reflects that for very particular Hamiltonians, our $L$ matrices need additional help to span all $1$-qubit operations. We include the proof in Appendix \ref{app:density}.

Now that we have shown density in $SL(2,\mathbb{C})$, as well as the fact that we can produce inverses of the gates in our generating set, our proof of yielding $\PP$ under postselection follows almost immediately. 
In particular, we will invoke the following theorem by Aharanov, Arad, Eban and Landau \cite{aharonov}:
\begin{theorem}[\cite{aharonov} Theorem 7.6, adapted to our case]\label{thm:sk}
There exists a constant $\epsilon_0 >0$ such that, for any $G=\{g_1 \ldots g_k\} \subset SL(2,\mathbb{C})$ which is an $\epsilon_0$-net over $B$, where $B$ is the set of operations in $SL(2,\mathbb{C})$ which are $2.1$-far from the identity (which in particular contains $SU(2)$), then for any unitary $U\in SU(2,\mathbb{C})$, there is an algorithm to find an $\epsilon$-approximation to $U$ using $\polylog(1/\epsilon)$ elements of $G$ and their inverses which runs in $\polylog(1/\epsilon)$ time.
\end{theorem}
In the above theorem, when we say an operation is ``$\epsilon$-far" from another, we are referring to the operator norm.

From this, we can immediately prove the main theorem. Suppose we wish to compute a language $L_0 \in \PP$, and we have a commuting Hamiltonian $H$ of the form promised in Theorem \ref{maintheorem}. By Aaronson's result that $\PP \subseteq \PostBQP$ \cite{postbqp}, there is an efficiently computable postselected quantum circuit $C$ composed of Hadamard and Toffoli gates which computes $L$. Additionally, by Claim \ref{density} there exists a finite set $G$ of products of $L$'s and $L^{-1}$'s which form an $\epsilon_0$-net over $B$ (which can be computed in finite time). Hence by Theorem \ref{thm:sk} there is a poly-time algorithm which expresses single-qubit gates as products of elements of $G$ to exponential accuracy. Likewise, since $H$ is entangling, we can generate some entangling two-qubit gate $g$, as well as its inverse $g^{-1}$ (by applying $-H$). Since $g$ and single-qubit gates are universal \cite{anygateuniversal}, by the usual Solovay--Kitaev theorem \cite{solovay}, we can express the circuit $C$ in terms of $g$, $g^{-1}$, and single-qubit gates to exponential accuracy with polynomial overhead. Combining these, we can express the circuit $C$ as a polynomial sized product of $g$'s, $g^{-1}$'s, $L$'s, and $L^{-1}$'s , which we can express as a $\PostIQPH$ circuit using the gadgets described previously. Hence this $\PostIQPH$ circuit decides the language $L_0$. 

Note that in this construction, it is crucial that we only ever need to invert a finite number of $L(t)$ matrices. This ensures that the size of the postselection gadgets involved to construct the $L^{-1}$ operations are upper bounded by a constant depending on the choice of $H$ only. Additionally, it is important that we can construct the $L^{-1}$ matrices exponential accuracy. This is crucial because in order to perform $\PostBQP$ under postselection, one needs to be able to simulate Aaronson's algorithm to exponential accuracy \footnote{This is because the algorithm postselects on an exponentially unlikely event, so to maintain polynomial accuracy after postselection we require exponential accuracy prior to postselection.}. Fortunately our construction allows us to simulate the algorithm to high accuracy, and hence these Hamiltonians can be used to sample from probability distributions which are not possible to simulate with a classical computer unless the polynomial hierarchy collapses.

This completes the proof in all cases except the exceptional case $H=X(\theta)\otimes X(\theta)$. This has a separate hardness of sampling result which was shown by Fefferman, Foss-Feig, and Gorshkov \cite{feffermanpersonal}. In particular, they showed the following:
\begin{theorem}[Fefferman et al. \cite{feffermanpersonal}]\label{feffermanlastcase}
If $H=X(\theta)\otimes X(\theta)$ for some $\theta$, then a $\BPP$ machine cannot weakly simulate \sampIQPH with any constant multiplicative error unless $\PH$ collapses to the third level.
\end{theorem}
Their proof makes use of that fact that using such Hamiltonians, for any matrix $A\in\{0,\pm1\}^n$, one can perform 
a unitary $U$ on a system of $O(n)$ qubits such that $\bra{1^n}U\ket{0^n} = k\left(\text{Perm}(A) + \epsilon\right)$,
where $k$ is independent of $A$ and exponentially small in $n$, $\text{Perm}(A)$ denotes the permanent of $A$, and $\epsilon$ is a term with norm $o(2^{-n})$. 
Note that Perm$(A)^2$ is $\#\P$-hard to compute with any constant multiplicative error \cite{bosonsampling}.
Therefore Theorem \ref{feffermanlastcase} immediately follows by the techniques of Aaronson and Arkhipov \cite{bosonsampling} - because if there were an efficient classical simulation of such circuits, then using approximate counting \cite{stockmeyer}, one could approximate $\text{Perm}(A)^2$ to multiplicative error $\left(1+\frac{1}{\poly(n)}\right)$ in $\BPP^\NP$. But $\BPP^\NP \subseteq \Delta_3$, so again by Toda's theorem \cite{Toda} this implies the collapse of $\PH$ to the third level.

This completes the last remaining case, and hence completes the proof.

\end{proof}

\section{Open Problems}
Our results leave a number of open problems.
\begin{enumerate} 
\item An interesting open problem is to classify \emph{all} Hamiltonians in terms of their computational power under this model. Childs et al. \cite{2qubithamiltonians} previously classified which two-qubit Hamiltonians can perform any unitary on two qubits. However, this does not classify which Hamiltonians are computationally universal for two reasons. First, as Childs et al. point out in their paper, it is possible that $H$ fails to generate all unitaries on two qubits, but does generate all unitaries on three qubits (i.e. adding ancillas helps one attain universality). It remains open to classify which two-qubit $H$ generate all unitaries on sufficiently large systems. Second, even if a Hamiltonian $H$ does not generate all unitaries, it is still possible that $H$ is computationally universal. For example, $H$ could be universal on an encoded subspace. Classifying which Hamiltonians are universal under an encoding seems to be a challenging task. We conjecture that the power of any two-qubit Hamiltonian obeys a dichotomy: either $H$ is efficiently classicaly simulable in this model, or it is universal under postselection and hence cannot be weakly simulated unless $\PH$ collapses. This is true of all known two-qubit Hamiltonians, and our classification proves this result rigorously in the case of commuting Hamiltonians.
\item  
In this paper we considered the power of quantum circuits with commuting Hamiltonians.
A more difficult related problem is classify the power of quantum circuits with commuting gate sets.
The challenge in solving this problem would be to classify when a discrete set of $L$'s generates a continuum of gates. 
There are some sufficient conditions under which this holds (see e.g. Aharonov et al.\ \cite{aharonov}, Corollary 9.1).
However, finding necessary and sufficient conditions under which a finite set of operators densely generates a continuous subgroup of $SL(2,\mathbb{C})$ seems very difficult, in part because there is no complete, explicit classification of discrete subgroups of $SL(2,\mathbb{C})$. 
Indeed, discrete subgroups of $SL(2,\mathbb{C})$ are related to the theory of M\"{o}bius transformations \cite{beardonsl2c}, where they are known as ``Kleinian subgroups," and they are the subject of a deep area of mathematical research.

\end{enumerate}

\section{Acknowledgements}
We thank Bill Fefferman, Michael Foss-Feig, and Alexey Gorshkov for allowing us to include a summary of their unpublished work \cite{feffermanpersonal}. 
We also thank Jacob Taylor for pointing us to a simplified proof of Claim \ref{clm:localdiag}, Michael Bremner for pointing us to references  \cite{shepherdIQPthesis,shepherdmatroids}, and Scott Aaronson, Joseph Fitzsimons and Ashley Montanaro for helpful discussions.
A.B. was supported in part by the National Science Foundation Graduate Research Fellowship under Grant No.\ 1122374, in part by the Center for Science of Information (CSoI), an NSF Science and Technology Center, under grant agreement CCF-0939370, and in part by Scott Aaronson's NSF Waterman award. 
L.M.~was supported by Singapore Ministry of Education (MOE) and National Research Foundation Singapore, as well as MOE Tier 3 Grant ``Random numbers from quantum processes'' (MOE2012-T3-1-009). 
X.Z. was supported by the MIT UROP program.

\appendix
\section*{Appendix}
\section{Commuting Hamiltonians are locally diagonalizable}\label{app:localdiag}

\newcommand{\braket}[2]{\langle#1|#2\rangle}
\newcommand{\proj}[1]{|#1\rangle\langle#1|}
\newcommand{\id}{I}
\renewcommand{\C}{\mathbb{C}}
\renewcommand{\R}{\mathbb{R}}

\newcommand{\Eq}[1]{Equation~(\ref{eq:#1})}
\newcommand{\EqRef}[1]{(\ref{eq:#1})}
\newcommand{\Lem}[1]{Lemma~\ref{lem:#1}}

To establish Claim \ref{clm:localdiag}, we prove the following stronger statement.

\begin{clm}
If $H$ is a two-qubit Hamiltonian and $[H\otimes \id, \id \otimes H]= 0$, then $(U\otimes U) H (U\otimes U)^\dagger$ is diagonal for some one-qubit unitary $U$.
\end{clm}

This is actually slightly stronger than Lemma 33 of \cite{cubitt}, which shows that if $[H\otimes \id, \id \otimes H]= [H\otimes \id, \id \otimes THT]= [THT\otimes \id, \id \otimes H]=0$, then $H$ is locally diagonalizable. Here we merely require that $[H\otimes \id, \id \otimes H]= 0$.

\begin{proof}
As a first step we expand $H$ in Pauli basis and let $\alpha_{AB}$ be the coefficient at $A\otimes B$ term for any $A,B\in\set{\id,X,Y,Z}$. Also, for all $A\in\set{\id,X,Y,Z}$, let 
\begin{equation}
  \vec{c}_A := (\alpha_{XA},\alpha_{YA},\alpha_{ZA})^T
  \quad \text{and} \quad
  \vec{r}_A := (\alpha_{AX},\alpha_{AY},\alpha_{AZ})^T.
\end{equation}
Given a vector $\vec{v} = (v_x,v_y,v_z)^T\in\R^3$, we adopt a commonly used notation and write $\vec{v}\cdot\vec{\sigma}$ to denote the linear combination $v_x X + v_y Y + v_z Z$.

Since $(H\otimes \id)(\id \otimes H) =(\id \otimes H)  (H\otimes \id)$, we know that both products must have the same expansion in Pauli basis. Let us fix $A,B\in \set{\id,X,Y,Z}$ and consider the terms of the form $A\otimes \underline{\phantom{A}} \otimes B$ in the Pauli expansion of each of the products. 

First, for $(H\otimes \id)(\id \otimes H)$ we notice that, when restricted to terms of the form $A\otimes \underline{\phantom{A}} \otimes B$, its Pauli expansion is given by
\begin{align}
  &\big( A \otimes (\alpha_{AI} \id + \vec{r}_A\cdot \vec{\sigma}) 	
  \otimes\id \big)
  \big( \id \otimes (\alpha_{IB} \id + \vec{c}_B\cdot \vec{\sigma}) 	
  \otimes B \big) = \\
  & A \otimes \big( \alpha_{AI}\alpha_{IB} \id + 
  (\alpha_{AI}\vec{c}_B + \alpha_{IB}\vec{r}_A) \cdot \vec{\sigma} +
  (\vec{r}_A \cdot \vec{\sigma})  (\vec{c}_B \cdot \vec{\sigma})
  \big) \otimes B = \\
  & A \otimes \big( 
  (\alpha_{AI}\alpha_{IB} + \vec{r}_A \cdot \vec{c}_B) \id + 
  (\alpha_{AI}\vec{c}_B + \alpha_{IB}\vec{r}_A + 
  i (\vec{r}_A \times \vec{c}_B)) \cdot \vec{\sigma}
  \big) \otimes B, \label{eq:Pauli1}
\end{align}
where we have applied the identity $(\vec{v}\cdot\vec{\sigma}) (\vec{w}\cdot\vec{\sigma}) = (\vec{v}\cdot \vec{w}) \id + i (\vec{v} \times \vec{w})\vec{\sigma}$ in the last step.

Next, we consider the product $(\id \otimes H)(H\otimes \id)$ and similarly obtain that, when restricted to terms of the form $A\otimes \underline{\phantom{A}} \otimes B$, its the Pauli expansion is given by
\begin{align}
  &\big( \id \otimes (\alpha_{IB} \id + \vec{c}_B\cdot \vec{\sigma}) 	
  \otimes B \big)
  \big( A \otimes (\alpha_{AI} \id + \vec{r}_A\cdot \vec{\sigma}) 	
  \otimes \id \big) = \\
  & A \otimes \big( 
  (\alpha_{AI}\alpha_{IB} + \vec{c}_B \cdot \vec{r}_A) \id + 
  (\alpha_{AI}\vec{c}_B + \alpha_{IB}\vec{r}_A + 
  i (\vec{c}_B \times \vec{r}_A)) \cdot \vec{\sigma}
  \big) \otimes B. \label{eq:Pauli2}
\end{align}
Since the coefficients in the Pauli expansions of $(H\otimes \id)(\id \otimes H)$ have to coincide with those in the expansion of $(\id \otimes H)(H\otimes \id)$, we know that the difference between expressions \EqRef{Pauli1} and \EqRef{Pauli2} equals zero. Considering the middle tensor and canceling some therms gives
\begin{equation}
  (\vec{r}_A \times \vec{c}_B) \cdot \vec{\sigma} = 
  (\vec{c}_B \times \vec{r}_A) \cdot \vec{\sigma}.
\end{equation}
Since $\vec{v} \times \vec{w} = - \vec{w} \times \vec{v}$, we obtain that $\vec{r}_A \times \vec{c}_B = 0$. This further implies that $\vec{r}_A$ and $\vec{c}_B$ are collinear, that is, $\dim(\spn\set{\vec{r}_A,\vec{c}_B}) \le 1$. Since we can choose arbitrary $A,B \in \set{I,X,Y,Z}$, it must be that all the vectors $\vec{r}_A$ and $\vec{c}_B$ must lie in the same one-dimensional subspace, i.e.,
\begin{equation}
  \dim\big( \spn\set[\big]{\vec{r}_A,\vec{c}_B 
  : A, B \in \set{I,X,Y,Z}} \big) \le 1.
\label{eq:rank1}
\end{equation}
Let us now consider a $3\times 3$ matrix $M$ whose rows and columns are indexed by Pauli matrices $X,Y$ and $Z$ and its entries are defined  via $M_{AB} = \alpha_{AB}$. Then the vectors $\vec{c}_A$ are the columns of $M$ and $\vec{r}_B$ are its rows. From \Eq{rank1}, we see that $M$ has rank at most one. Moreover, the row and column spaces of $M$ must coincide as 
\begin{equation}
  \spn\big(\set{\vec{r}_X,\vec{r}_Y,\vec{r}_Z}\big) = 
  \spn\big(\set{\vec{c}_X,\vec{c}_Y,\vec{c}_Z}\big).
\end{equation}
These two observations imply that $M = \vec{v}\vec{v}^T$ for some $\vec{v} \in \R^3$. So we can express our Hamiltonian $H$~as
\begin{equation}
  H  =  \alpha_{II} \id\otimes \id + 
		(a\vec{v} \cdot \vec{\sigma}) \otimes \id +
		\id \otimes (b\vec{v} \cdot \vec{\sigma}) +
		(\vec{v} \cdot \vec{\sigma}) \otimes (\vec{v} \cdot \vec{\sigma}),
\label{eq:HamForm}
\end{equation}
where $a,b\in\R$ are such that $\vec{r_I} = a \vec{v}$ and $\vec{c_I} = b \vec{v}$. If we pick a unitary $U$ that diagonalizes $\vec{v}\cdot\vec{\sigma}$, then from \Eq{HamForm} we see that $U\otimes U$ diagonalizes our Hamiltonian $H$. This concludes the proof.
\end{proof}

\section{Inverting $L$ matrices using postselection gadgets} \label{app:inverses}

We now prove Claim \ref{clm:inverses}.

\begin{proof}

We will need two additional gadgets for our construction. First, consider a modification of the gadget for $L(t)$, where we start the qubit in the $\ket{1}$ state and postselect on the $\ket{1}$ state:
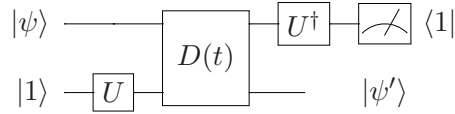
\[\Qcircuit @C=1em @R=1em {
\lstick{\ket{\psi}} & \qw & \multigate{1}{D(t)} & \gate{U^\dagger} & \meter & \bra{1}   \\
\lstick{\ket{1}} & \gate{U} & \ghost{D(t)} & \qw & \ket{\psi'} 
}
\]
By a direct calculation, one can show the linear transformation performed on $\ket{\psi}$ is given by
\[M(t) = \frac{1}{|\alpha| |\beta| \sqrt{e^{-2it} - e^{2it}}}  \left( \begin{matrix} |\beta|^2 e^{ia't}& -\alpha \beta^* e^{it} \\ -\alpha^* \beta e^{i} & |\alpha|^2 e^{id't} \end{matrix} \right)\]
This is tantalizingly close to the inverse of $L$, which is 
\[L(t) ^{-1} = \frac{1}{|\alpha| |\beta|\sqrt{e^{-2it} - e^{2it}}  }  \left( \begin{matrix} |\beta|^2 e^{id't}& -\alpha \beta^* e^{it} \\ -\alpha^* \beta e^{it} & |\alpha|^2 e^{ia't} \end{matrix} \right)\]

The only thing that is off is that the phase of the upper left and bottom right entries are incorrect. 
We now break into three cases to describe how to correct the phases in each. (Recall that $d'=-2-a'$ as our without loss of generality our Hamiltonian is traceless).

Case 1: $a'=d'=-1$
In this case we already have $M(t)=L^{-1}(t)$, so we have found the inverse. 

Case 2: $a'=1, d'=-3$ OR $a'=-3, d'=1$

We will prove the case $a'=1$; an analogous proof holds for $a'=-3$.

To correct the phases in $M(t)$, we need to introduce an additional gadget:
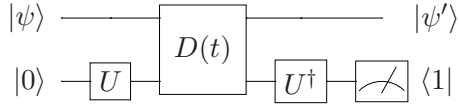
\[\Qcircuit @C=1em @R=1em {
\lstick{\ket{\psi}} & \qw & \multigate{1}{D(t)} &  \qw & \qw & \ket{\psi'} \\
\lstick{\ket{0}} & \gate{U} & \ghost{D(t)} & \gate{U^\dagger}  &\meter & \bra{1}  
}
\]
In other words, instead of using the gate in a teleportation-like protocol, we instead use it to apply phases to $\ket{\psi'}$. This gate performs the following transformation on the input state:
\[
N(t) = \frac{1}{\sqrt{( e^{it}-e^{ia't} )( e^{id't } -e^{it} ) }} \left(\begin{matrix} e^{it}-e^{ia't}   & 0 \\ 0 & e^{id't}  -e^{it}  \end{matrix}\right)
\]

In the case that $a'=1$, this gadget becomes singular, and hence it performs the operation 
\[\left(\begin{matrix} 0   & 0 \\ 0 & 1  \end{matrix}\right)\]
In other words, this gadget postselects the qubit involved on the state $\ket{1}$.
This holds in particular for $t=\pi/4$. (In fact it holds for any $t$ such that $e^{-3it}\neq e^{it}$, in which case it becomes undefined). 

By composing $N(\pi/4)$ with other gadgets, this now empowers us to create gadgets in which we postselect on $\ket{1}$ on lines which do not end in $U^\dagger$. For instance, we can create the following gadget:
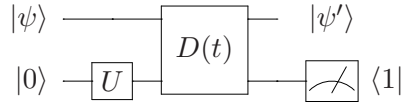
\[\Qcircuit @C=1em @R=1em {
\lstick{\ket{\psi}} & \qw & \multigate{1}{D(t)} & \qw &  \ket{\psi'}   \\
\lstick{\ket{0}} & \gate{U} & \ghost{D(t)}      &\qw & \meter & \bra{1} 
}
\]
Which one can easily check is equivalent to the following circuit, which maintains the property that every line begins and ends with $U$ and $U^\dagger$.
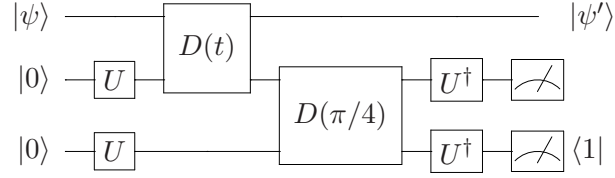
\[\Qcircuit @C=1em @R=1em {
\lstick{\ket{\psi}} & \qw       & \multigate{1}{D(t)} & \qw                          & \qw & \qw &  \ket{\psi'}   \\
\lstick{\ket{0}}    & \gate{U} & \ghost{D(t)}          &\multigate{1}{D(\pi/4)} & \gate{U^\dagger} &  \meter  \\
\lstick{\ket{0}} & \gate{U} & \qw                         & \ghost{D(\pi/4)} & \gate{U^\dagger} & \meter & \bra{1} \
}
\]
This is simply composing the gadget with $N(\pi/4)$. (Here the output of the middle qubit is an independent sample from measuring the state $U^\dagger \ket{1}$ in the computational basis).

This gadget performs the following operation on $\ket{\psi}$:
\[
P(t) \propto  \left(\begin{matrix} e^{it} &0 \\ 0 & e^{-3it} \end{matrix}\right) \propto \left(\begin{matrix} e^{2it} &0 \\ 0 & e^{-2it}\end{matrix}\right)
\]
In other words, the matrix $P(t)$ is a phase gate by phase $\theta=2t$. 

The construction of arbitrary phase gates suffices to correct the diagonal phases of $M(t)$, because for any matrix $\left(\begin{matrix} a & b \\ c &d\end{matrix}\right)$ we have that
\[ \left(\begin{matrix} e^{i\theta/2} & 0 \\ 0 & e^{-i\theta/2}\end{matrix}\right)\left(\begin{matrix} a & b \\ c &d\end{matrix}\right) \left(\begin{matrix} e^{i\theta/2} & 0 \\ 0 & e^{-i\theta/2}\end{matrix}\right)= \left(\begin{matrix} ae^{i\theta} & b \\ c &de^{-i\theta}\end{matrix}\right)\]
Hence by choosing $\theta=(d'-a')t$, and multiplying $M(t)$ by this matrix on both sides, we obtain $L^{-1}(t)$ as desired. Clearly this construction is efficient, i.e. the postselection gadget is of constant size, and one can efficiently compute the times to run the Hamiltonians in the gadget to high precision. This completes the proof.

Case 3: $a'\neq \pm 1, -3$

To correct the phases in $M(t)$, we need to consider the same gadget $N(t)$ which we used in Case 2:
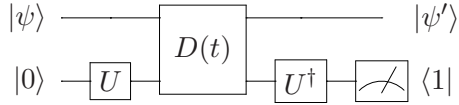
\[\Qcircuit @C=1em @R=1em {
\lstick{\ket{\psi}} & \qw & \multigate{1}{D(t)} &  \qw & \qw & \ket{\psi'} \\
\lstick{\ket{0}} & \gate{U} & \ghost{D(t)} & \gate{U^\dagger}  &\meter & \bra{1}  
}
\]
In other words, instead of using the gate in a teleportation-like protocol, we instead use it to apply phases to $\ket{\psi'}$. This gate performs the following transformation on the input state:
\[
N(t) = \frac{1}{\sqrt{( e^{it}-e^{ia't} )( e^{id't } -e^{it} ) }} \left(\begin{matrix} e^{it}-e^{ia't}   & 0 \\ 0 & e^{id't}  -e^{it}  \end{matrix}\right)
\]
Since $N$ is a diagonal matrix, the only physical quantity that matters is the ratio $r(t)$ of its two entries, which is a complex number given by
\[r(t) = \frac{  e^{it} -e^{ia't} }{e^{id't} -e^{it}}.\]
If $r(t)$ takes on a certain value, then it immediately follows that $N(t) =\pm \left(\begin{smallmatrix} \sqrt{r} & 0 \\ 0 & \sqrt{r}^{-1} \end{smallmatrix}\right)$, because of our normalization. 
Furthermore, if we compose $N(s)N(t)$, then the ratio of the resulting diagonal matrix is $r(s)r(t)$. 
Note the $\pm1$ term is an irrelevant global phase, so we omit it in the further calculations. 

We will now show that for any complex phase $e^{i\theta}$, where $\theta\neq 0, \pi$, there exists a finite set of times $t_1,t_2,...t_k$, $s_1,s_2,...s_{k'}$ such that 
\[
N(t_1)N(t_2)...N(t_k) N(s_1)N(s_2)...N(s_{k'}) = \left(\begin{matrix} e^{i\theta/2} & 0 \\ 0 & e^{-i\theta/2}\end{matrix}\right)
\]

As previously mentioned in Case 2, the construction of such matrices suffices to correct the diagonal phases of $M(t)$, because for any matrix $\left(\begin{matrix} a & b \\ c &d\end{matrix}\right)$ we have that
\[ \left(\begin{matrix} e^{i\theta/2} & 0 \\ 0 & e^{-i\theta/2}\end{matrix}\right)\left(\begin{matrix} a & b \\ c &d\end{matrix}\right) \left(\begin{matrix} e^{i\theta/2} & 0 \\ 0 & e^{-i\theta/2}\end{matrix}\right)= \left(\begin{matrix} ae^{i\theta} & b \\ c &de^{-i\theta}\end{matrix}\right)\]
Hence by choosing $\theta=(d'-a')t$, and multiplying $M(t)$ by this matrix on both sides, we obtain $L^{-1}(t)$ as desired and this will complete the proof.

To prove this, we will prove two separate facts. First, we will show that given $\theta$, there exists a sequence $t_1,t_2,...t_k$ such that $N(t_1)N(t_2)...N(t_k) =\left(\begin{matrix} ce^{i\theta/2} & 0 \\ 0 & \frac{1}{c}e^{-i\theta/2}\end{matrix}\right)$  for some $c\in\mathbb{R}^+$. 
Next, we will show that for any $c\in\mathbb{R}$, there exists a sequence $s_1,s_2,...s_{k'}$ of times such that 
$N(s_1)N(s_2)...N(s_{k'}) = \left(\begin{matrix} 1/c & 0 \\ 0 & c\end{matrix}\right)$. 
Together these imply the claim.

Moreover, we will show this construction is efficiently computable. 
More specifically, suppose you want to find invert $L$.
The for each $L$ the size of the postselection gadget required to invert $L$ is of constant size.
Additionally, the amount of computational time required to compute the values of $t_i$ and $s_i$ to ensure that we find $L^{-1}$ to accuracy $\epsilon$ scales as poly$\log(k_L * 1/\epsilon)$, where $k_L$ is a constant which depends on $L$. 
Furthermore, the times $t_i$ and $s_i$ are upper bounded by a constant which only depends on the value of $a'$. 
For convenience we will refer to these properties as "the construction is efficiently computable."

At first glance it might sound like this definition of ``efficiently computable" is too weak, because the inverses of arbitrary $L$ matrices might require large postselection gadgets, or might require a long time to compute the values of the $t_i$ and $s_i$ to sufficient accuracy. 
However, later in our construction we will use the fact that for any fixed Hamiltonian $H$, we will only need to invert a fixed number of $L$ matrices.
Hence for fixed $H$, the size of the postselection gadgets which appear in our circuit will be upper bounded by a constant depending on $H$ only, but not on the size of the problem we are solving under postselection. Furthermore, for fixed $H$, we can compute the times $t_i,s_i$ required to invert the relevant $L$ matrices to exponential accuracy in poly$\log(1/epsilon)$ time (where a hidden constant $k_L$ depending on $L$ has been absorbed into the big-O notation).

\begin{clm}\label{clm:phase}For any $\theta\in(0,2\pi)$, there exists a sequence $t_1,t_2,...t_k$ such that $N(t_1)N(t_2)...N(t_k) =\left(\begin{matrix} ce^{i\theta/2} & 0 \\ 0 & e^{-u\theta/2}/c\end{matrix}\right)$  for some $c\in\mathbb{R}^+$. Furthermore, this construction is computationally efficient.
\end{clm}
\begin{proof}
To see this, consider the expression for the ratio
\[r(t) = \frac{  e^{it} -e^{ia't} }{e^{id't} -e^{it}} = -\frac{  1 -e^{i(a'-1)t} }{1-e^{i(d'-1)t}}.\]
Let $\text{Phase}(c)$ denote the phase of $c$ modulo $2\pi$. Then by direct calculation we have that
\begin{align*}
\text{Phase}(r(t)) &= \pi+\text{Phase}\left(\frac{  1 -e^{i(a'-1)t} }{1-e^{i(-3-a')t}}\right) \\
&= \pi + \text{Phase}\left(  1 -e^{i(a'-1)t} \right) - \text{Phase}\left(1 -e^{i(-3-a')t} \right) \\
&=\pi + \text{Phase}\left(  1 -e^{i(a'-1)t} \right) + \text{Phase}\left(1 -e^{i(3+a')t} \right) \\
&=\left( \pi +\left( \frac{(a'-1)t}{2}\mod \pi\right) +  \left(\frac{(3+a')t}{2}\mod \pi\right) \right) \mod 2\pi \\
&=\left( \pi +\left(t'\mod \pi\right) +  \left(Rt'\mod \pi\right) \right) \mod 2\pi
\end{align*}
Where $t'=(a'-1)t/2$ and $R=\frac{(3+a')}{(1-a')}$. 
Since we are in the case that $a'\neq \pm1, -3$, we are promised that $R$ is well-defined and $R\neq 0, 1$. 
Also note that we cannot have that $R=-1$ because this would imply $3=-1$, a contradiction.

Suppose $R>0$ (an analogous proof holds for $R<0$). Then for $t'\in[0, \min(\pi, \pi/R)]$, we know that $\text{Phase}(r(t')) = \pi + (R+1)t'$, because in this range $t'$ is sufficiently small such that both $t'\mod\pi=t'$ and $Rt' \mod \pi = Rt'$. 
Hence using $t'$ in this interval, we can achieve any phase in $(\pi,\pi+s)$ where $s=(R+1)\min(\pi,\pi/R)$. For any $R\neq 0,-1$ this range is of constant size. 
Thus by multiplying together $1/s$  phases in the range $(\pi,\pi+s)$, one can achieve any phase in $(0,2\pi)$, as desired.

Note that this construction is manifestly efficient; the $t_i$'s are upper bounded by a constant $\min(\pi,\pi/R)$ which is a function of $H$ only, and computing them to polynomially many digits requires polynomial time, as it just requires simple addition.

\end{proof}

\begin{clm}\label{clm:norm}For any $c\in\mathbb{R}^+-\{1\}$, there exists a finite sequence $s_1,s_2,...s_k$ such that 
\[N(s_1)N(s_2)...N(s_k) = \left(\begin{matrix} 1/c & 0 \\ 0 & c\end{matrix}\right)\]
\end{clm}
\begin{proof}
Consider products of matrices of the form $N(s)N(-s)$ for $s\in\mathbb{R}+$. 
Let $f(s)=r(s)r(-s)$. One can check by direct calculation that
\[f(s) = \frac{1-\cos((1-a')s)}{1-\cos((3+a')s)}\]
In other words, the product of the ratios is real and positive, hence the resulting matrix $N(s)N(-s)$ is of the form $\left(\begin{matrix} 1/\ell & 0 \\ 0 & \ell\end{matrix}\right)$ for some $\ell\in \mathbb{R}^+$. 
Note since we are in the case $a'\neq \pm1, -3$ this ratio is well-defined.

If we redefine $s'=s/(1-a')$, and set $R=(1-a')/(3+a')$, then this ratio becomes
\[\frac{1-\cos s'}{1 - \cos Rs'}
\]
We know $R\neq 0, 1$ because we have $a'\neq \pm1, 3$, and furthermore $R\neq -1$ as well, since this would imply $1=-3$, a contradiction.

For clarity of explanation assume $R>0$; an analogous proof holds for the case $R<0$.
 
Next we claim that the range of $f(s)$ as $s$ varies over $R$ includes the interval \[(\min(R^{-2}, R^2),\max(R^{-2},R^2)).\] 
Since $R\neq1$ this is an interval of constant size around $1$. 
To see this, we will break into two cases.

First, assume $R>1$. Consider the value of this function when $s'\in(0, \pi/R)$. The function $f(s')$ in continuous in this range. Additionally $\lim_{s'\rightarrow 0} f(s') = 1/R^2$ by L'H\^{o}pital's rule, and $\lim_{s'\rightarrow \pi/R} = +\infty$. 
Hence the range of $f$ covers $(R^{-2}, +\infty) = (\min(R^{-2},R^2),+\infty)$ by the mean value theorem.

Next, assume $0<R<1$. Now consider the value of the function when $s'\in(0,\pi)$. Again the function is continuous in this range, and we have $\lim_{s'\rightarrow 0} f(s') = 1/R^2$ by L'H\^{o}pital's rule, and $\lim_{s'\rightarrow \pi} =0$. 
Hence the range of $f$ covers $(0,R^{-2}) = (0,\max(R^{-2},R^2))$ by the mean value theorem.

Hence in either case, by choosing an appropriate value of $s'$, we can set $f(s)$ to be any real value in a finite-length interval containing $1$. 
Hence for any target ratio $c^2\in\mathbb{R}^+$, one can take a finite product of $O(\log(c))$ values of $f(s)$ such that $f(s_1)f(s_2)...f(s_k)=c^2$. This implies the claim.

Note that this construction is efficient. First, the times $s_i$ are upper bounded by $\min(\pi,\pi/R)$, which is a constant which depends on the Hamiltonian $H$ only. Second, to compute each individual time $s_i$, one simply needs to solve the problem 
\[
\frac{1-\cos s'}{1 - \cos Rs'} = k
\]
For some $k\in(\min(R^{-2}, R^2),\max(R^{-2},R^2))$ and $s'$ in $(0,\min(\pi,\pi/R))$. In the region of $s$ where the value of this function is between $\min(R^{-2}, R^2)$ and $\max(R^{-2},R^2))$ , the derivatives of this function are bounded by a function of $R$ only. 
Furthermore, the derivatives of these terms are computable to accuracy $\epsilon$ in time poly$\log(1/\epsilon)$ time using the Taylor series for sine and cosine. 
Hence Newton's method can be used to solve this problem, and will achieve quadratic convergence, i.e. for each step you run Newton's method, the error is squared, and the number of digits of accuracy achieved doubles. 
Hence one can compute each time $t_i$ to accuracy $\epsilon$ in poly$\log(1/\epsilon)$ time as desired.
Furthermore, since inverting any particular $L$ only requires inverting some fixed $c\in\mathbb{R}^+$ using Claim \ref{clm:norm}, an error $\epsilon$ in an individual $N(s_i)$ matrices contributes $c\epsilon$ error to the operator norm\footnote{This is because for non-unitary matrices, the norm of the singular values are not one. 
Hence when considering the product $AB$, where $\lambda_{max}$ is the largest singular value of $A$, an $\epsilon$ error in $B$ will induce an $\lambda_{max}\epsilon$ error in AB.} of $N(s_1)...N(s_k)$, and hence $c\epsilon$ error to the operator norm of $L^{-1}$. 
Hence this construction is ``computationally efficient" for each fixed $L$ as defined previously.

\end{proof}

This completes the proof in Case 3 and hence the entire proof.

\end{proof}

\section{Showing density in $SL(2,\mathbb{C})$}\label{app:density}

We now prove Claim \ref{density}.

\begin{proof}[Proof of Claim \ref{density}]
To show that $S=SL(2,\mathbb{C})$, we will first show that $S$ is a group, and then show $S$ is a Lie group.
\begin{clm} $S$ is a group.
\end{clm}
\begin{proof}
Clearly, if we only took finite products of these elements, the resulting set of matrices would be a group, because we have the inverses of every element in the generating set. So what we need to show is that taking the closure of this set of matrices still yields a group. To see this, suppose that some element $s\in S\subseteq SL(2,\mathbb{C})$  is the limit of a sequence $L_1, L_2, \ldots$ where each $L_i$ is a finite product of element of the form $L(D(t_1,t_2))$, and $\lim_{i\rightarrow \infty} L_i =s$. Now consider the sequence $L_1^{-1},L_2^{-1}, \ldots$. We claim that $\lim_{i \rightarrow \infty} L_i^{-1} = s^{-1}$. To see this, simply note that for a $2\times 2$ matrix $\left(\begin{smallmatrix}a & b \\ c & d\end{smallmatrix}\right)\in SL(2,\mathbb{C})$, its inverse is given by $\left(\begin{smallmatrix}d & -b \\ -c & a\end{smallmatrix}\right)$. Since the limit point $s$ exists in $SL(2,\mathbb{C})$, the limit of each matrix entry of the $L_i$'s must converge as well to the entries of $s$. Hence the entries of the sequence $L_i^{-1}$ converges to the entries of $s^{-1}$. 
\end{proof}
Note that it is critical that we've taken the closure in $SL(2,\mathbb{C})$; if we took the closure in the set of $2\times 2$ complex matrices, this would not necessarily be true.

We have now established that $S$ is a group. Furthermore, $S$ is a closed subgroup of $SL(2,\mathbb{C})$ by construction, and $SL(2,\mathbb{C})$ is a Lie group. We now invoke a well-known theorem from Lie theory.
\begin{theorem}[Cartan's Theorem \cite{Cartan1952} or the Closed Subgroup Theorem] Any closed subgroup of a Lie group is a Lie group.
\end{theorem}
\begin{corollary}$S$ is a Lie group. \end{corollary}

Now that we know $S$ is a Lie group, we can use facts from Lie theory to show $S=SL(2,\mathbb{C})$.
We will summarize the basics here, but a more complete treatment can be found in e.g.\ \cite{Hall} or a more advanced textbook on Lie groups.

 A Lie group is a continuous manifold which is also a group, for which the group operations are smooth. 
In this work we will only consider matrix groups, i.e. continuous groups of complex matrices. 
For any Lie group $G$, one can define the Lie algebra of $G$, denoted $\Lie(G)$, to be the tangent space to the group $G$ at the identity.
 More concretely, suppose that you have a smooth path $\gamma(t):\mathbb{R} \rightarrow G\subseteq GL(n,\mathbb{C})$ in $G$, such that $\gamma(0)=I$. 
Then the matrix $\frac{\partial}{\partial t}\gamma(t)\Bigr|_{t=0}$ belongs to the tangent space of $G$ at the identity. 
One can show that $\Lie(G)$ obeys the following properties \cite{Hall}:
\begin{enumerate}
\item $\mathfrak{g}$ is a real vector space, i.e. $g_1, g_2\in\mathfrak{g} \Rightarrow ag_2+bg_2\in\mathfrak{g}$ for any $a,b\in\mathbb{R}$.
\item $\mathfrak{g}$ is closed under commutators, i.e.  $g_1 , g_2\in\mathfrak{g} \Rightarrow [g_1,g_2]\triangleq g_1g_2-g_2g_1 \in\mathfrak{g}$ for any $a,b\in\mathbb{R}$.
\item \label{lie:exp} Let $\exp(A) = I + A + \frac{A^2}{2} + \frac{A^3}{6} +\ldots + \frac{A^n}{n!}+\ldots$. Then we have that for all $g\in\mathfrak{g}$, $\exp(g)\in G$. In other words, the function $\exp$ maps from the Lie algebra into the Lie group. 
\item \label{lie:conj} $\mathfrak{g}$ is closed under taking commutators with the group $G$. That is, for any  $G_1 \in G$ and $g\in \mathfrak{g}$, we have $G_1gG_1^{-1} \in \mathfrak{g}$. 
\end{enumerate}

To show that $S=SL(2,\mathbb{C})$, we will consider $\mathfrak{g} \triangleq \Lie(S)$. We will then show that $\mathfrak{g} = \mathfrak{sl}(2,\mathbb{C})$, which is the Lie algebra of $SL(2,\mathbb{C})$, which consists of all traceless two by two complex matrices. By property \ref{lie:exp}, this implies that $\exp(\mathfrak{sl}(2,\mathbb{C})) \subseteq S$. 
From this, we will leverage the following fact:
\begin{clm} $\exp(\mathfrak{sl}(2,\mathbb{C}))$ is dense in $SL(2,\mathbb{C})$.
\end{clm}
\begin{proof} It is well known \cite{Hall} that  $\exp(\mathfrak{sl}(2,\mathbb{C}))$ contains all matrices in $SL(2,\mathbb{C})$ except matrices $A$ for which $\Tr (A)=-2$ and $A\neq -I$. This implies the claim. \qedhere
\end{proof}

Hence to prove Claim \ref{density}, it suffices to prove the following claim:

\begin{clm}\label{clm:liealg} $\mathfrak{g}\triangleq \Lie(S)$ spans $\mathfrak{sl}(2,\mathbb{C})$, i.e. all $2\times 2$ traceless matrices.
\end{clm}
\begin{proof}

Consider elements of the form 
\[
M(t,s) \triangleq L(t) L(s)^{-1}.
\]
As $t,s$ vary over $(0,\pi) \cup (\pi,2\pi)$, these form continuous paths within $S$. 
In particular, at the point where $s=t$, this path passes through the identity. 
Now consider 
\[
g(v) \triangleq \frac{\partial}{\partial t} \left[M(t,s)\right] \Bigr|_{\substack{s=t=v}}
\]
These are tangent vectors to paths in $S$, evaluated as they pass through the identity. 
Hence we have that $g(v) \in \mathfrak{g}$ for all $v \in  (0,\pi) \cup (\pi,2\pi)$. 
By direct calculation, one can show that
\[
g(v) = -\frac{1}{2\sin(2v)}\left(\begin{matrix} (a'+1)e^{-2iv} & \frac{\alpha}{\beta} (1-a') e^{i(1+a')v} \\ \frac{\beta}{\alpha} (3+a')e^{i(-1-a')v} & -(a'+1)e^{-2iv}  \end{matrix}\right)
\]
where we have simplified using the fact that $d'=-2-a'$. 

We will now break into cases to show that these matrices span the entire Lie algebra. We begin with the generic case and then give the special cases. In the special cases, we will also add additional postselection gadgets to our model in order to get single-qubit transformations which span all traceless matrices. 
The gadgets introduced are inherently closed under taking inverses. So this simply reflects that for very particular Hamiltonians, our $L$ matrices need additional help to span all $1$-qubit operations.

\textbf{Case 1}: $a'\neq \pm1,-3$

In this case all of the entries of $g(v)$ are non-zero. 
\[
g(v) = -\frac{1}{2\sin(2v)}\left(\begin{matrix} (a'+1)e^{-2iv} & \frac{\alpha}{\beta} (1-a') e^{i(1+a')v} \\ \frac{\beta}{\alpha} (3+a')e^{i(-1-a')v} & -(a'+1)e^{-2iv}  \end{matrix}\right)
\]
We can therefore rewrite $g(v)$ with four non-zero parameters $k_1\in\mathbb{R}$, $k_2,k_3\in\mathbb{C}$, and using a new parameter $v'=-2v$:
\[
g(v) \propto \left(\begin{matrix} e^{v'} & k_2 e^{ik_1v'} \\ k_3e^{-ik_1v'} & -e^{iv'}  \end{matrix}\right)
\]
Here we omit real coefficients as the Lie algebra is closed under scalar multiplication by $\mathbb{R}$. The fact that $a'\neq \pm1, -3$ also implies that $k_4\neq \pm1$

Now consider the value of $g(v')$ for small values of $v'$. In particular, pick a $\theta << 1$. Then we have that 
\[
g(\pm \theta) \propto \left(\begin{matrix} (A\pm Bi) & k_2(C\pm Di) \\  k_3(C \mp Di) & -(A\pm Bi)  \end{matrix}\right)
\]
for some nonzero real coefficients $A,B,C,D\in\mathbb{R}$.
Taking the sum and difference of these matrices, we see the following are elements of the Lie algebra:
\[ \begin{matrix} \left(\begin{matrix} A & k_2C \\  k_3C & -A  \end{matrix}\right)
&
\left(\begin{matrix} Bi & k_2Di \\  -k_3Di & -Bi  \end{matrix}\right)
\end{matrix}
\]
Likewise, by considering taking the sum and difference of $g(\pm 2\theta)$, we  get there exist nonzero $A',B',C',D' \in\mathbb{R}$ such that the lie algebra contains.
\[ \begin{matrix} \left(\begin{matrix} A' & k_2C' \\  k_3C' & -A'  \end{matrix}\right)
&
\left(\begin{matrix} B'i & k_2D'i \\  -k_3D'i & -B'i  \end{matrix}\right)
\end{matrix}
\]
Furthermore, since sine and cosine are nonlinear, and $k_1\neq 0,\pm1$, the vectors $(A,C)$ and $(A',C')$ are linearly independent. Likewise the vectors $(B,D)$ and $(B',D')$ are linearly independent. Hence by taking linear combinations of these matrices, we have that any matrix of the form
\[ 
\left(\begin{matrix} E & k_2F \\  k_3F^* & -E  \end{matrix}\right)
\]
is in the Lie algebra for any $E,F\in\mathbb{C}$. Hence our Lie algebra spans at least these two complex dimensions. Now we take the closure of such matrices under commutators. Suppose $A,B,C,D\in\mathbb{C}$. We have that
\[
\left[ \left(\begin{matrix} A & k_2B \\  k_3B^* & -A  \end{matrix}\right), \left(\begin{matrix} C & k_2D \\  k_3D^* & C  \end{matrix}\right) \right] = \left(\begin{matrix} k_2k_3(BD^*-B^*D) & 2k_2(AD-BC)\\ 2k_3(B^*C-AD^*)  &  k_2k_3(B^*D-BD^*) \end{matrix}\right)
\]
Since we previously showed all traceless diagonal matrices are in the Lie algebra, this implies the following matrices are in the Lie algebra:
\[
 \left(\begin{matrix} 0 & 2k_2(AD-BC)\\ 2k_3(B^*C-AD^*)  &  0 \end{matrix}\right)
\]
By setting $A,D,B,C$ such that $(AD-BC)^*\neq (B^*C-AD^*)$, we can see that these matrices span the remaining two real dimensional space of off-diagonal matrices. Hence our Lie algebra spans all traceless matrices. This completes the proof in Case 1.

\textbf{Case 2}: $a'=1$ or $a'=-3$

We will prove the claim for $a'=1$; an analogous proof holds for $a'=-3$. (These are the Hamiltonians $\text{diag}(1,1,1,-3)$ and $\text{diag}(-3,1,1,1)$, which are identical except the role of 0 and 1 is switched.).

In this case we have that
\[
g(v) \propto \left(\begin{matrix} e^{-2iv} &  0 \\ 2\frac{\beta}{\alpha} e^{-2iv} & -e^{-2iv}  \end{matrix}\right)
\]
By evaluating $g(v)$ at $\pm\theta$ and $\pm2\theta$ for some small value of $\theta$, by the same arguments put forth in Case 1, these matrices span the space of matrices of the form
\[ \left(\begin{matrix}A+Bi &  0 \\ 2\frac{\beta}{\alpha} (A+Bi) & -A-Bi  \end{matrix}\right)\]
Where $A,B\in\mathbb{R}$ are arbitrary real parameters.

We will now use another postselection gadget, which is inherently closed under taking inverses, to boost the span of the algebra to all of $\mathfrak{sl}(2,\mathbb{C})$. This is the same gadget which appears in the construction of $L^{-1}$ in appendix \ref{app:inverses}. 
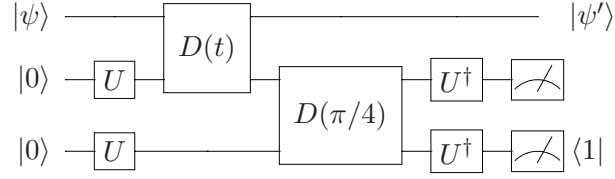
\[\Qcircuit @C=1em @R=1em {
\lstick{\ket{\psi}} & \qw       & \multigate{1}{D(t)} & \qw                          & \qw & \qw &  \ket{\psi'}   \\
\lstick{\ket{0}}    & \gate{U} & \ghost{D(t)}          &\multigate{1}{D(\pi/4)} & \gate{U^\dagger} &  \meter  \\
\lstick{\ket{0}} & \gate{U} & \qw                         & \ghost{D(\pi/4)} & \gate{U^\dagger} & \meter & \bra{1} \
}
\]
This gadget performs the operation 
\[
P(t) \propto  \left(\begin{matrix} e^{it} &0 \\ 0 & e^{-3it} \end{matrix}\right) \propto \left(\begin{matrix} e^{2it} &0 \\ 0 & e^{-2it}\end{matrix}\right)
\]
Hence its Lie algebra spans the space of traceless diagonal imaginary matrices. Combining this with the previous result, we see the Lie algebra now spans the space
\[ \left(\begin{matrix}A+Bi &  0 \\ 2\frac{\beta}{\alpha} (A+Ci) & -A-Bi  \end{matrix}\right)\]
Where $A,B,C\in\mathbb{R}$ are arbitrary real parameters.

Now consider taking commutators of such matrices; one can easily see that for $A,B,C,D,E,F\in\mathbb{R}$, 
\[
\left[ \left(\begin{matrix} A+Bi & 0 \\  2\frac{\beta}{\alpha}(A+Ci) & -A-Bi  \end{matrix}\right), \left(\begin{matrix} D+Ei & 0 \\  2\frac{\beta}{\alpha}(D+Fi) & -D-Ei  \end{matrix}\right) \right] = \left(\begin{matrix}0 & 0\\ 4\frac{\beta}{\alpha}(A+Ci)(D+Ei)  &  0 \end{matrix}\right)
\]
Hence by appropriate choice of $A,C,D,E$ these commutators span all complex values in the lower left hand corner. So our Lie algebra now spans
\[ \left(\begin{matrix}A+Bi &  0 \\  C+Di & -A-Bi  \end{matrix}\right)\]
Where $A,B,C,D\in\mathbb{R}$ are arbitrary real parameters.
In other words we span all traceless lower triangular matrices.

Next we will use the fact that the Lie algebra is closed under conjugation by the group. 
Therefore it must contain all elements of the form
\[ L(t) \left(\begin{matrix}A &  0 \\  B & -A  \end{matrix}\right) L^{-1}(t)\]
where $A,B$ are now complex parameters

Since we already span lower triangular matrices, the only relevant entry of the above matrix is the upper-right entry, as we can zero out the other entries by adding lower triangular matrices. 
This upper left entry is proportional to
\[i\left(-2\alpha\beta^*|\alpha|^2  e^{2it}A - \alpha^2\beta^{*2}e^{2it}B\right)\]
Since $\alpha$ and $\beta$ are non-zero, and setting $B=0$, we can see that by choosing $A$ we can set this value to be any complex number. Hence our Lie algebra must span
\[ L(t) \left(\begin{matrix}A &  C \\  B & -A  \end{matrix}\right) L^{-1}(t)\]
Where $A,B,C \in\mathbb{C}$, that is all of $\mathfrak{sl}(2,\mathbb{C})$, as desired. This completes the proof of Claim 2.

\textbf{Case 3:} $a'=-1$

In this case we have that 
\[g(v) = -\frac{1}{\sin(2v)}\left(\begin{matrix} 0 & \frac{\alpha}{\beta}  \\ \frac{\beta}{\alpha}  & 0 \end{matrix}\right)\]
Thus the matrices $g(v)$ span a one-dimensional space. 
Since the Lie algebra is closed under scalar multiplication by reals, the factor of $\frac{-1}{\sin(2v)}$ out front is irrelevant, and we will drop real prefactors in future calculations.

We will now use the fact the Lie algebra is closed under conjugation by the group. 
Consider matrices of the form
\[
T(s,v)=L(s)g(v)L(s)^{-1} \propto i \left(\begin{matrix} |\beta|^4 - |\alpha|^4&  |\alpha|^4 \frac{\alpha}{\beta} e^{-2is} -\alpha\beta^*|\beta|^2e^{2is} \\ \frac{\beta}{\alpha} |\beta|^4e^{-2is} - \alpha^*\beta|\alpha|^2e^{2is}& |\alpha|^4-|\beta|^4\end{matrix}\right)
\]
where the proportionality is over real scalar multiples. 
Here we have simplified using the fact we are in the case $a'=d'=-1$. 
This is well defined for any $s$ and $v$ which are not integer multiples of $\pi$.

Now we break into two subcases:

\textbf{Subcase A}: $|\alpha|^2\neq |\beta|^2$

In this case, the matrix $T(s,v)$ has a nonzero entry on the diagonals. Hence the matrix $T(s,v)$ has the form
\[
T(s,v) \propto i \left(\begin{matrix} k_1 &  k_2 e^{-2is} -k_3e^{2is} \\ k_4e^{-2is} - k_5e^{2is}& -k_1\end{matrix}\right)
\]
Where $k_1\in\mathbb{R}$ is nonzero, $k_2,k_3,k_4,k_5\in\mathbb{C}$ are nonzero. 
One can easily check that the constraint $|\alpha|^2\neq|\beta|^2$ further implies that $k_2,k_3,k_4,k_5$ have four distinct values, i.e. $k_i\neq k_j$ for any $i\neq j$, $i,j\geq 2$. 
For instance, to see that $k_2\neq k_3$, note that if $k_2=k_3$ then $|\alpha|^4\frac{\alpha}{\beta} = \alpha\beta^*|\beta|^2$, which implies $|\alpha|^4=|\beta|^4$, a contradiction.

Furthermore, one can show that there cannot exist a constant\footnote{ If this were the case,   the matrices $T(s,v)$ would only span matrices of the form $\left(\begin{matrix}Ai &  B+Ci \\  K(B+Ci) & -Ai  \end{matrix}\right) $. Fortunately this does not happen in this case.} $K$ such that $k_2=Kk_4$ and $k_3=Kk_5$, because this would imply $|K|=|\frac{\alpha}{\beta}|^6=|\frac{\alpha}{\beta}|^2$ which is a contradiction if $|\alpha|\neq |\beta|$. Hence the matrices $T(s,v)$ span matrices of the form
\[ 
\left(\begin{matrix}Ai &  B+Ci \\  D+Ei & -Ai  \end{matrix}\right)
\]
where $A,B,C,D,E\in\mathbb{R}$ are arbitrary real parameters. 
Now taking the closure of such matrices under commutators, one can easily see this spans all traceless matrices. 
Hence the Lie algebra spans $\mathfrak{sl}(2,\mathbb{C})$ as desired.

\textbf{Subcase B}: $|\alpha|^2= |\beta|^2=1/2$

In this case the Hamiltonians generated are of the form $X(\theta)\otimes X(\theta)$, so are not covered in the scope of this theorem. 
Note that the Lie alebgra of the $L$ gadgets here only span a two dimensional subspace of the form
\[\left(\begin{matrix}0& e^{-i\theta} (A+Bi) \\ e^{i\theta}(A+Bi) & 0\end{matrix}\right)\]
where $A,B\in\mathbb{R}$. This is closed under conjugation and does not span $\mathfrak{sl}(2,\mathbb{C})$.

\end{proof}
\end{proof}

\section{Proof of postselected universality when $b\neq c$}\label{app:edgecase}

Here we consider the postselected universality of circuits with entangling Hamiltonians for which $H\neq THT$. The proof in this case will follow analogously to the main proof. Furthermore, the construction of the inverse gadgets will have a much cleaner construction than the case $H=THT$.

Suppose we have a commuting Hamiltonian $H$ such that $H\neq THT$. By Claim \ref{clm:localdiag}, we know that $H = (U\otimes U) \diag (a,b,c,d) (U^\dagger \otimes U^\dagger)$ for some one-qubit unitary $U=\left(\begin{smallmatrix} \alpha & -\beta^* \\ \beta & \alpha^*\end{smallmatrix}\right)$ and some real parameters $a,b,c,d$.
The trace of $H$ contributes an irrelevant global phase to the unitary operator it generates, so without loss of generality we can assume $H$ is traceless, i.e., $a+b+c+d=0$. Since $H\neq THT$ we have $b\neq c$. As before, the fact $H$ can generate entanglement starting from the computational basis implies $\alpha\neq0,\beta\neq0$, and $b+c\neq 0$. 

Now consider the Hamiltonians 
\[
	H_1 = \frac{1}{c^2-b^2}(cH_{12} - bH_{21}), \quad H_2 = \frac{1}{b^2-c^2}(bH_{12} - cH_{21})
\]
Since we can apply both $H$, $-H$, $THT$, and $-THT$, this allows us to apply $H_1$ and $H_2$ for independent amouts of time.  Let $V(t_1,t_2)$ be the two-qubit unitary we obtain from running $H_1$ for time $t_1\in \mathbb{R}$ and $H_2$ for time $t_2\in \mathbb{R}$. We have
\[
V(t_1,t_2) = e^{it_1H_1}e^{it_2H_2} = (U^{\otimes 2}) D(t_1,t_2) ({U^\dagger}^{\otimes 2}),
\]
where $D(t_1,t_2)\triangleq \diag(e^{ia'(t_1+t_2)},e^{it_1},e^{it_2},e^{id'(t_1+t_2)})$. 

Now following our previous proof, we consider the following postselection gadget: 
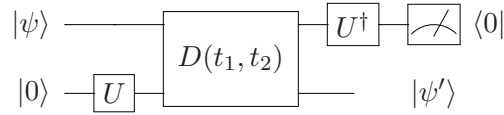
\[\Qcircuit @C=1em @R=1em {
\lstick{\ket{\psi}} & \qw & \multigate{1}{D(t_1,t_2)} & \gate{U^\dagger} & \meter & \bra{0}   \\
\lstick{\ket{0}} & \gate{U} & \ghost{D(t_1,t_2)} & \qw & \ket{\psi'} 
}
\]
This performs the following transformation on the input state:
\[L(t_1,t_2) = \frac{1}{|\alpha||\beta|\sqrt{\left( e^{-i(t_1+t_2)} - e^{i(t_1+t_2)}\right)}}  \left( \begin{matrix} |\alpha|^2 e^{ia'(t_1 + t_2)}& \alpha \beta^* e^{it_2} \\ \alpha^* \beta e^{it_1} & |\beta|^2 e^{id'( t_1 + t_2)} \end{matrix} \right).\]

As before, this is a non-unitary transformation, and hence it is unclear how to invert $L$. Fortunately, when $H\neq THT$ we have the freedom to apply $H_1$ and $H_2$ for separate times, and this allows us to make a much simpler postselecting gadget to invert $L$, as follows:

\begin{clm} \label{clm:inversesedge} Given $L(t_1,t_2)$, where $t_i\in (0,\pi)\cup(\pi,2\pi)$, it is possible to construct $L(t_1,t_2)^{-1}$ by introducing three postselections into the circuit. Furthermore, this construction is efficiently computable in the manner described above.
\end{clm}
\begin{proof}
We will need two additional gadgets for our construction. First, consider a modification of the gadget for $L(t_1, t_2)$, where we start the qubit in the $\ket{1}$ state and postselect on the $\ket{1}$ state:
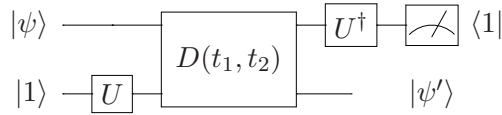
\[\Qcircuit @C=1em @R=1em {
\lstick{\ket{\psi}} & \qw & \multigate{1}{D(t_1,t_2)} & \gate{U^\dagger} & \meter & \bra{1}   \\
\lstick{\ket{1}} & \gate{U} & \ghost{D(t_1,t_2)} & \qw & \ket{\psi'} 
}
\]
By a direct calculation, one can show the linear transformation performed on $\ket{\psi}$ is given by
\[M(t_1,t_2) = \frac{1}{|\alpha| |\beta| \sqrt{\left( e^{-i(t_1+t_2)} - e^{i(t_1+t_2)}\right)}}  \left( \begin{matrix} |\beta|^2 e^{ia'(t_1 + t_2)}& -\alpha \beta^* e^{it_2} \\ -\alpha^* \beta e^{it_1} & |\alpha|^2 e^{id'( t_1 +  t_2)} \end{matrix} \right)\]
This is tantalizingly close to the inverse of $L$, which is 
\[L(t_1,t_2) ^{-1} = \frac{1}{|\alpha| |\beta|\sqrt{\left( e^{-i(t_1+t_2)} - e^{i(t_1+t_2)}\right)}}  \left( \begin{matrix} |\beta|^2 e^{id'(t_1 + t_2)}& -\alpha \beta^* e^{it_2} \\ -\alpha^* \beta e^{it_1} & |\alpha|^2 e^{ia'( t_1 +  t_2)} \end{matrix} \right)\]
The only thing that is off is that the phase of the upper left and bottom right entries are incorrect. To correct these phases, we need to introduce another gadget:
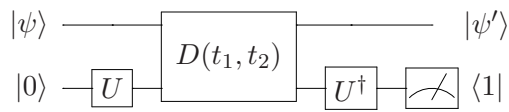
\[\Qcircuit @C=1em @R=1em {
\lstick{\ket{\psi}} & \qw & \multigate{1}{D(t_1,t_2)} &  \qw & \qw & \ket{\psi'} \\
\lstick{\ket{0}} & \gate{U} & \ghost{D(t_1,t_2)} & \gate{U^\dagger}  &\meter & \bra{1}  
}
\]
In other words, instead of using the gate in a teleportation-like protocol, we instead use it to apply phases to $\ket{\psi'}$. This gate performs the following transformation on the input state:
\[
N(t_1,t_2) = \frac{1}{\sqrt{( e^{it_1}-e^{ia'(t_1 +  t_2)} )( e^{id'(t_1 + t_2) } -e^{it_2} ) }} \left(\begin{matrix} e^{it_1}-e^{ia'(t_1 +  t_2)}   & 0 \\ 0 & e^{id'(t_1 +  t_2)}  -e^{it_2}  \end{matrix}\right)
\]
Since $N$ is a diagonal matrix, the only physical quantity that matters is the ratio $r(t_1,t_2)$ of its two entries, which is a complex number given by
\[r(t_1,t_2) = \frac{  e^{it_1} -e^{ia'(t_1 +  t_2)} }{e^{id'(t_1 +  t_2)} -e^{it_2}}.\]
If $r = r(t_1,t_2)$ takes on a certain value, then it immediately follows that $N(t_1,t_2) = \left(\begin{smallmatrix} \sqrt{r} & 0 \\ 0 & \sqrt{r}^{-1} \end{smallmatrix}\right)$, because of our normalization. 

We will now show that by setting $t_1$ and $t_2$, we can choose $r(t_1,t_2)$ to be any complex phase $e^{i\theta}$ that we like. 
In fact, if $\frac{a'}{d'}$ is irrational, one can also show that one can choose $t_1,t_2$ to approximate any complex number; however, this will not be necessary for our construction, so we omit this here.
\begin{clm} For any $\theta \in (0,2\pi)$, there exist $t_1,t_2 \in \mathbb{R}$ such that $r(t_1,t_2) = e^{i\theta}$.
\end{clm}
\begin{proof}
Set $t_1=\theta$ and $t_2=-\theta$. We immediately have
\[
r(\theta,-\theta) = \frac{e^{i\theta} - 1}{1 - e^{-i\theta}} = \frac{e^{i\theta} - 1}{e^{-i\theta} (e^{i\theta} - 1)} = e^{i\theta}.
\]
Note that this only works if $e^{i\theta} \neq 1$ - this is why we have omitted $\theta=0$ from our range of $\theta$. In other words, this gadget can be used to perform any diagonal matrix other than the identity. 
\end{proof}

Putting this all together, we now show how to invert $L(t_1,t_2)$. 
Set $s_1 = i(d'(t_1+t_2) - a'(t_1+t_2)) $ and $s_2 =-s_1 $. Then we have\footnote{This is possible as long as $e^{i(d'(t_1+t_2) - a'(t_1+t_2))} \neq 1$. If this quantity is one, then $L(t_1,t_2)^{-1} = M(t_1,t_2)$, so no additional gadgets are necessary to obtain inverses.}  
\[N(s_1,s_2) = \left(\begin{matrix} e^{\frac{i}{2}(d'(t_1+t_2) - a'(t_1+t_2))}& 0 \\ 0 & e^{-\frac{i}{2}(d'(t_1+t_2) - a'(t_1+t_2))}\end{matrix}\right)
\]
Now one can easily check that
\[
L(t_1,t_2)^{-1} = N(s_1,s_2) M(t_1,t_2) N(s_1,s_2)
\]
And therefore the following gadget performs $L(t_1,t_2)^{-1}$:
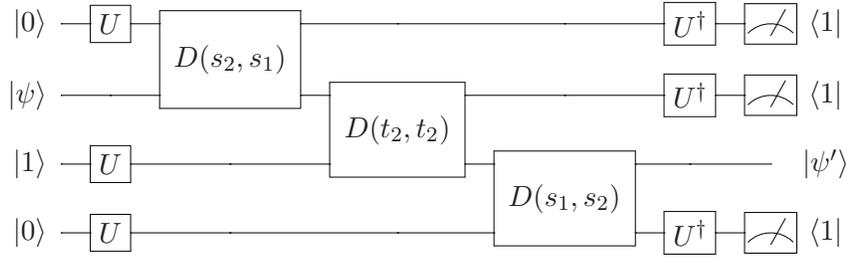
\[\Qcircuit @C=1em @R=1em {
\lstick{\ket{0}}    & \gate{U} & \multigate{1}{D(s_2,s_1)}   & \qw                                      & \qw                                    &\gate{U^\dagger}      &\meter   & \bra{1} &  \\
\lstick{\ket{\psi}} & \qw       & \ghost{D(s_2,s_1)}             &  \multigate{1}{D(t_2,t_2)}      &  \qw                                    & \gate{U^\dagger}  &\meter & \bra{1} \\
\lstick{\ket{1}} & \gate{U} & \qw                                      &   \ghost{D(t_1,t_2)}                  & \multigate{1}{D(s_1,s_2)}  & \qw                      & \qw     & \ket{\psi'}          \\
\lstick{\ket{0}} & \gate{U} & \qw                                      & \qw                                         & \ghost{D(s_1,s_2)}             & \gate{U^\dagger}    & \meter & \bra{1} } 
\]
(Note that $s_1$ and $s_2$ are switched in the first diagonal matrix, as we have switched the usual order of the qubits.)
 
Hence using these postselection gadgets, we can generate not only $L(t_1,t_2)$, but also its inverse. 
Furthermore, this construction is manifestly efficient, since $s_1$ and $s_2$ are efficiently computable given $t_1$ and $t_2$.  
\end{proof}

We can therefore apply both $L(t_1,t_2)$ and $L(t_1,t_2)^{-1}$ in our postselected circuits. This once again allows us to apply Lie theory to determine which subset of transformations can be applied by taking products of $L$ matrices. Following our proof of the main theorem, we now show the Lie algebra of the $L$ matrices spans $\mathfrak{sl}(2,\mathbb{C})$.  This completes the proof of postselected universality in this case in analogy with the main theorem.

\begin{clm} The Lie algebra of the $L$ matrices spans  $\mathfrak{sl}(2,\mathbb{C})$ in the case where $T\neq THT$.
\end{clm}
\begin{proof}

Consider elements of the form 
\[
M(t_1,t_2,s_1,s_2) \triangleq L(D(t_1,t_2)) L(D(s_1,s_2))^{-1}.
\]
As $t_1,t_2,s_1,s_2$ vary over the set $ \left\{t_1,t_2 : t_1+t_2 \in (0,\pi)\cup(\pi,2\pi) \right\} \times \left\{s_1,s_2 : s_1+s_2 \in (0,\pi)\cup(\pi,2\pi) \right\}$, these form continuous paths within $S$. In particular, at the point where $s_1=t_1$ and $s_2=t_2$, this path passes through the identity. Now consider 
\[
g(v_1,v_2) \triangleq \frac{\partial}{\partial t_1} \left[M(t_1,t_2,s_1,s_2)\right] \Bigr|_{\substack{s_1=t_1=v_1\\s_2=t_2=v_2}}
\]
and
\[
h(v_1,v_2) \triangleq \frac{\partial}{\partial t_2} \left[M(t_1,t_2,s_1,s_2)\right] \Bigr|_{\substack{s_1=t_1=v_1\\s_2=t_2=v_2}}.
\]
These are tangent vectors to paths in $S$, evaluated as they pass through the identity. Hence we have that $g(v_1,v_2)$ and $h(v_1,v_2) \in \mathfrak{g}$ for all $v_1,v_2 \in  \left\{v_1,v_2 : v_1+v_2 \in (0,\pi)\cup(\pi,2\pi) \right\}$. By direct calculation, one can show that
\[
g(v_1,v_2) = -\frac{1}{2\sin(v_1+v_2)}\left(\begin{matrix} a'e^{-i(v_1+v_2)} + \cos(v_1+v_2) & -\frac{\alpha}{\beta} a' e^{i(a'v_1+(a'+1)v_2)} \\ \frac{\beta}{\alpha} (2+a')e^{i((d'+1)v_1 + d' v_2)} & -a'e^{-i(v_1+v_2)} - \cos(v_1+v_2) \end{matrix}\right)
\]
and
\[
h(v_1,v_2) = -\frac{1}{2\sin(v_1+v_2)}\left(\begin{matrix} a'e^{-i(v_1+v_2)} - i \sin(v_1+v_2) & \frac{\alpha}{\beta} (1-a')e^{i(a'v_1) + (a'+1)v_2} \\ \frac{\beta}{\alpha}(1+a')e^{i((d'+1)v_1 + d'v_2)} & -a'e^{-i(v_1+v_2)} + i\sin(t+1+v_2) \end{matrix} \right)
\]
where we have simplified using the fact that $d'=-1-a'$. Now suppose that we evaluate these matrices at the points where $v_1=\theta$ and $v_2 = \frac{\pi}{2} - \theta$ for some real parameter $\theta$; this ensures that $v_1,v_2$ are in the allowed set, and simplifies the above expressions to
\begin{align*}
g(\theta)& = -\frac{1}{2}\left(\begin{matrix} -a'i & -\frac{\alpha}{\beta} a' e^{i(-\theta + (a'+1)\frac{\pi}{2})} \\ \frac{\beta}{\alpha}(2+a')e^{i(\theta + d'\frac{\pi}{2})} & a'i \end{matrix}\right) \\
&= -\frac{1}{2} \left(\begin{matrix} -a'i & -\frac{\alpha}{\beta}a'e^{i\theta'} \\ \frac{\beta}{\alpha}(2+a')e^{-i\theta'} & a'i \end{matrix}\right),
\end{align*}
here we define $\theta' = -\theta + (a'+1)\frac{\pi}{2}$; this follows from the fact that $d'=-1-a'$. 
Likewise, we can consider $h(v_1,v_2)$ evaluated when $v_1=\theta$ and $v_2 = \frac{\pi}{2} - \theta$; this evaluates to
\begin{align*}
h(\theta) &= -\frac{1}{2}\left(\begin{matrix} -ia' - i  & \frac{\alpha}{\beta} (1-a')e^{i(-\theta + (a'+1)\frac{\pi}{2})} \\ \frac{\beta}{\alpha}(1+a')e^{i(\theta + d'\frac{\pi}{2})} & ia' + i \end{matrix} \right)
\\ &= -\frac{1}{2} \left(\begin{matrix} -i(a' +1)  & \frac{\alpha}{\beta} (1-a')e^{i\theta'} \\ \frac{\beta}{\alpha}(1+a')e^{-i\theta'} & i(a' +1) \end{matrix} \right).
\end{align*} 
By setting the value of $\theta$ in the range $[0,2\pi)$, we can select any values of $\theta'$ we like; hence we will work with $\theta'$ from this point forward.

For now we will assume that $a'\neq 0$ and $a'\neq 1$; we will handle the cases $a'=0$ and $a'=-1$ separately. The proof of the general case is the most difficult one.

\textbf{Case 1:} $a'\neq 0$ and $a'\neq -1$.

We know that $g(\theta') \in \mathfrak{g}$ and $h(\theta') \in \mathfrak{g}$ . Furthermore, since $\mathfrak{g}$ is a real Lie algebra, it is closed as a vector space over $\mathbb{R}$. Hence we must also have that 
\[ j(\theta_1,\theta_2) \triangleq -2\left( \frac{1}{a'+1}h(\theta_2) - \frac{1}{a'}g(\theta_1)\right)  = \left(\begin{matrix} 0 & \frac{\alpha}{\beta} \left( \frac{1-a'}{1+a'}e^{i\theta_2} + e^{i\theta_1}   \right) \\ \frac{\beta}{\alpha}\left( e^{-i\theta_2} - \frac{2+a'}{a'}e^{-i\theta_1}   \right) & 0 \end{matrix} \right) \in \mathfrak{g}\]
Where we have used the assumption that $a' \neq 0$ and $a' \neq -1$. We will now show that as we vary $\theta_1$ and $\theta_2$, these elements $j(\theta_1,\theta_2)$ span all two by two matrices of the form $\left(\begin{smallmatrix}0 & c_1 \\ c_2 & 0 \end{smallmatrix}\right)$, where $c_1,c_2 \in \mathbb{C}$.

To prove this, we will break into two subcases. For convenience, define
\[
	k = \frac{a'-1}{a'+1}.
\]

\textbf{Subcase A:} $a'>0$, i.e., $-1< k < 1$.

In this subcase, consider the matrices 
\begin{equation} \label{eq:span1} \frac{-a'(1+a')}{4} \left[ j\left( \arcsin k, \frac{\pi}{2}\right) + j\left( \pi - \arcsin k, \frac{\pi}{2}\right) \right]= \left(\begin{matrix}0 & 0 \\ \frac{\beta}{\alpha}i & 0\end{matrix}\right) \end{equation}
\begin{equation} \label{eq:span2} \frac{1+a'}{4\sqrt{a'}} \left[ j\left( \arcsin k, \frac{\pi}{2}\right) - j\left( \pi - \arcsin k, \frac{\pi}{2}\right) \right]= \left(\begin{matrix}0 & \frac{\alpha}{\beta} \\ -\frac{\beta}{\alpha} \frac{2+a'}{a'}& 0\end{matrix}\right) \end{equation}
and
\begin{equation}\label{eq:span3} \frac{a'(1+a')}{4} \left[ j\left( \arccos k, 0\right) + j\left( - \arccos k, 0 \right) \right]= \left(\begin{matrix}0 & 0 \\ \frac{\beta}{\alpha}& 0\end{matrix}\right) \end{equation}
\begin{equation}\label{eq:span4}  \frac{1+a'}{4\sqrt{a'}} \left[ j\left( \arccos k, 0\right) - j\left( - \arccos k, 0 \right) \right]= \left(\begin{matrix}0 & \frac{\alpha}{\beta}i \\ \frac{\beta}{\alpha} \frac{2+a'}{a'}i & 0\end{matrix}\right).\end{equation}

These are well-defined as we have $a'>0$ in this case. Clearly matrices (\ref{eq:span1}) and (\ref{eq:span3}) span the space of all matrices with a single complex entry in the bottom left hand corner. Hence, when combined with matrices (\ref{eq:span2}) and (\ref{eq:span4}), they clearly span the space of all matrices with complex entries in the off diagonal elements.

\textbf{Subcase B:} $a'<0$ and $a'\neq -1$, i.e., $-1< 1/k <1$

This subcase follows similarly; consider the matrices
\begin{equation} \label{eq:span5} \frac{a'(1-a')}{4} \left[ j\left( \frac{\pi}{2}, \arcsin \frac{1}{k} \right) + j\left(\frac{\pi}{2}, \pi - \arcsin\frac{1}{k}  \right) \right]= \left(\begin{matrix}0 & 0 \\ \frac{\beta}{\alpha}i & 0\end{matrix}\right) \end{equation}

\begin{equation} \label{eq:span6} \frac{1+a'}{4\sqrt{-a'}} \left[ j\left( \frac{\pi}{2}, \arcsin\frac{1}{k}\right) - j\left(\frac{\pi}{2}, \pi - \arcsin\frac{1}{k}  \right) \right]= \left(\begin{matrix}0 & \frac{\alpha}{\beta} \\ -\frac{\beta}{\alpha} \frac{1+a'}{1-a'}& 0\end{matrix}\right) \end{equation}

and
\begin{equation}\label{eq:span7} \frac{-a'(1-a')}{4} \left[ j\left( 0, \arccos\frac{1}{k}\right) + j\left( 0, - \arccos\frac{1}{k} \right) \right]= \left(\begin{matrix}0 & 0 \\ \frac{\beta}{\alpha}& 0\end{matrix}\right) \end{equation}

\begin{equation}\label{eq:span8} \frac{1+a'}{4\sqrt{-a'}} \left[ j\left( 0, \arccos\frac{1}{k}\right) - j\left( 0, - \arccos\frac{1}{k} \right) \right]= \left(\begin{matrix}0 & \frac{\alpha}{\beta} i  \\ \frac{\beta}{\alpha}\frac{1+a'}{1-a'}i & 0\end{matrix}\right). \end{equation}

These are well-defined as we have $a'<0$ in this case, as well as $a'\neq-1$. Again, clearly we have that (\ref{eq:span5}) and (\ref{eq:span7}) span all matrices with a single complex entry in the bottom left of the matrix. Hence, adding in (\ref{eq:span6}) and (\ref{eq:span8}), we span all off-diagonal complex matrices, which is what we wanted to show.

In either subcase, our $j$ matrices span all matrices of the form
\[\left(\begin{matrix}0 & A+Bi \\ C+Di & 0\end{matrix}\right)\] where $A,B,C,D \in\mathbb{R}$. Additionally, our $g$ and $h$ matrices are also in $\mathfrak{g}$, and clearly combining these with the $j$ matrices increases the span to 
\[\left(\begin{matrix}Ei & A+Bi \\ C+Di & -Ei\end{matrix}\right)\] where $A,B,C,D,E \in\mathbb{R}$. This is a five-dimensional subspace of $\mathfrak{sl}(2,\mathbb{C})$. Now to show that we can span all 6 dimensions of $\mathfrak{sl}(2,\mathbb{C})$, we invoke the fact that $\mathfrak{g}$ is closed under commutation, so $\mathfrak{g}$ contains $\left[\left(\begin{smallmatrix}0 & 1 \\ 0 & 0 \end{smallmatrix}\right)  , \left(\begin{smallmatrix}0 & 0 \\ 1 & 0 \end{smallmatrix}\right)  \right] = \left(\begin{smallmatrix}1 & 0 \\ 0 & -1 \end{smallmatrix}\right)$. Hence $\mathfrak{g}$ must include all matrices of the form
\[\left(\begin{matrix}F+Ei & A+Bi \\ C+Di & -F-Ei\end{matrix}\right)\] 
where $A,B,C,D,E,F \in \mathbb{R}$. In other words, $\mathfrak{g}=\mathfrak{sl}(2,\mathbb{C})$.

We've now shown Claim \ref{clm:liealg} in the case where $a'\neq 0$ and $a'\neq -1$. We now prove the claim in these remaining two cases.

\textbf{Case 2:} $a'=0$.

In this case we have
\[g(\theta')  = -\frac{1}{2} \left(\begin{matrix} 0 & 0 \\ \frac{\beta}{\alpha}2e^{-i\theta'} &0 \end{matrix}\right)\]
As $\theta$ varies these matrices clearly span all matrices a single complex number in the bottom left entry. Now in this case we also have that
\[
h(\theta')= -\frac{1}{2} \left(\begin{matrix} -i  & \frac{\alpha}{\beta} e^{i\theta'} \\ \frac{\beta}{\alpha}e^{-i\theta'} & i \end{matrix} \right)
\]
Since $\mathfrak{g}$ is closed under addition and scalar multiplication by $\mathbb{R}$, and applying

\[h(\theta') - h(\theta'') =  -\frac{1}{2} \left(\begin{matrix} 0  & \frac{\alpha}{\beta} (e^{i\theta'}-e^{i\theta''}) \\ \frac{\beta}{\alpha} (e^{-i\theta'} -e^{i\theta''})& 0 \end{matrix} \right)\in \mathfrak{g}\]
Now adding in multiples of $g$, we have that $\mathfrak{g}$ contains matrices of the form 
\[ \left(\begin{matrix} 0  & \frac{\alpha}{\beta} (e^{i\theta'}-e^{i\theta''}) \\  0 & 0 \end{matrix} \right)
\]
which clearly span all matrices with a complex entry in the upper right corner. Hence we span all off-diagonal matrices. Now adding in $h(\theta)$ for any $\theta$, we span all matrices of the form $\left(\begin{matrix}Ei & A+Bi \\ C+Di & -Ei\end{matrix}\right)$ where $A,B,C,D,E \in\mathbb{R}$. As discussed in Case 1, by taking the closure of these under commutation we have that $\mathfrak{g} = \mathfrak{sl}(2\mathbb{C})$ as desired, which completes the proof of Case 2.

\textbf{Case 3:} $a'=-1$

This case follows very similarly to Case 2. When $a'=-1$ we have that
\[h(\theta')= -\frac{1}{2} \left(\begin{matrix}0  & \frac{\alpha}{\beta} 2e^{i\theta'} \\ 0 & 0 \end{matrix} \right)\]
which clearly span all complex matrices with a single entry in the upper right corner. In this case, we also have that
\[g(\theta')= -\frac{1}{2} \left(\begin{matrix} i & -\frac{\alpha}{\beta}-e^{i\theta'} \\ \frac{\beta}{\alpha}e^{-i\theta'} & -i \end{matrix}\right),\]
By considering the difference  $g(\theta')- g(\theta'')$, and noting that we already span matrices with a single entry in the upper right corner, this shows that we span all off-diagonal matrices. Now adding in $g(\theta')$ for any $\theta'$ we see that we span all matrices of the form $\left(\begin{matrix}Ei & A+Bi \\ C+Di & -Ei\end{matrix}\right)$ where $A,B,C,D,E \in\mathbb{R}$. As discussed in Case 1, by taking the closure of these under commutation we have that $\mathfrak{g} = \mathfrak{sl}(2, \mathbb{C})$ as desired. This completes the proof of Case 3, hence the proof of the claim.
\end{proof}

\end{document}